\documentclass[journal,draftclsnofoot,onecolumn,12pt,twoside]{IEEEtran}

\usepackage{graphicx}
\usepackage{algorithm}
\usepackage{algorithmic}

\usepackage{booktabs}

\usepackage[colorlinks,
linkcolor=black,
anchorcolor=black,
citecolor=black,
urlcolor = blue
]{hyperref}

\usepackage{amssymb}
\usepackage{subfigure}
\usepackage{multicol}
\usepackage{multirow}
\usepackage{amsmath}
\usepackage{bm}
\usepackage{cite}
\usepackage{gensymb}
\usepackage{amsfonts}
\usepackage{mathrsfs}
\usepackage{amsmath}
\usepackage{color}
\usepackage{balance}
\usepackage{bm}
\usepackage{setspace}

\usepackage{stfloats}
\usepackage{float}
\usepackage{epstopdf}

\usepackage{array}
\newcommand{\PreserveBackslash}[1]{\let\temp=\\#1\let\\=\temp}
\newcolumntype{C}[1]{>{\PreserveBackslash\centering}p{#1}}
\usepackage[flushleft]{threeparttable}
\usepackage{stfloats}
\usepackage{mathrsfs}
\usepackage{threeparttable}
\usepackage{setspace}

\newtheorem{lemma}{Lemma}

\newtheorem{proposition}{\it Proposition}

\begin{document}

\bibliographystyle{IEEEtran} 

\title{Optical Intelligent Reflecting Surface \\Assisted MIMO VLC: Channel Modeling and Capacity Characterization}

\author{Shiyuan Sun, Weidong Mei, \emph{Member, IEEE}, Fang Yang, \emph{Senior Member, IEEE}, Nan An, Jian Song, \emph{Fellow, IEEE}, and Rui Zhang, \emph{Fellow, IEEE}	

\vspace{-1.8cm}

\begin{spacing}{0.3}
	\thanks{
		
		\scriptsize 
		
		S. Sun, F. Yang, and N. An are with the Department of Electronic Engineering, Tsinghua University, Beijing 100084, P. R. China, and also with the Key Laboratory of Digital TV System of Shenzhen City, Research Institute of Tsinghua University in Shenzhen, Shenzhen 518057, P. R. China.
		\textit{(Corresponding author: Fang Yang)}

		J. Song is with the Department of Electronic Engineering, Tsinghua University, Beijing 100084, P. R. China, and also with the Shenzhen International Graduate School, Tsinghua University, Shenzhen 518055, P. R. China.
		
		W. Mei is with the Department of Electrical and Computer Engineering, National University of Singapore, Singapore 117583.

		R. Zhang is with The Chinese University of Hong Kong, Shenzhen, and Shenzhen Research Institute of Big Data, Shenzhen, China 518172.
		He is also with the Department of Electrical and Computer Engineering, National University of Singapore, Singapore 117583.
	}
\end{spacing}
}
\newenvironment{thisnote}{\par\color{blue}}{\par}

\maketitle
\begin{abstract}
\vspace{-0.2cm}
Although the multi-antenna or so-called multiple-input multiple-output (MIMO) transmission has been the enabling technology for the past generations of radio-frequency (RF)-based wireless communication systems, its application to the visible light communication (VLC) still faces a critical challenge as the MIMO spatial multiplexing gain can be hardly attained in VLC channels due to their strong spatial correlation.
In this paper, we tackle this problem by deploying the optical intelligent reflecting surface (OIRS) in the environment to boost the capacity of MIMO VLC.
Firstly, based on the extremely near-field channel condition in VLC, we propose a new channel model for OIRS-assisted MIMO VLC and reveal its peculiar ``no crosstalk'' property, where the OIRS reflecting elements can be respectively configured to align with one pair of transmitter and receiver antennas without causing crosstalk to each other.
Next, we characterize the OIRS-assisted MIMO VLC capacities under different practical power constraints and then proceed to maximize them by jointly optimizing the OIRS element alignment and transmitter emission power.
In particular, for optimizing the OIRS element alignment, we propose two algorithms, namely, location-aided interior-point algorithm and log-det-based alternating optimization algorithm, to balance the performance versus complexity trade-off; while the optimal transmitter emission power is derived in closed form.
Numerical results are provided, which validate the capacity improvement of OIRS-assisted MIMO VLC against the VLC without OIRS and demonstrate the superior performance of the proposed algorithms compared to baseline schemes.
\end{abstract}


\begin{IEEEkeywords}
\vspace{-0.3cm}
Visible light communication (VLC),  multiple-input multiple-output (MIMO), optical intelligent reflecting surface (OIRS), capacity maximization, channel modeling.
\end{IEEEkeywords}

\IEEEpeerreviewmaketitle


\vspace{-0.5cm}
\section{Introduction}
\begingroup
\allowdisplaybreaks

\vspace{-0.1cm}
\IEEEPARstart{I}{n} the past several decades, wireless communication systems have been developed towards using higher frequency and broader bandwidth, installing more transceiver antennas, and adopting more efficient resource allocation schemes~\cite{goldsmith2005wireless}.
Compared to the spectrums used in the previous-generation (i.e.  2G-4G) wireless  systems, the spectrum allocated for today's 5G systems has migrated from the sub-6 GHz  frequency band to the higher millimeter wave (mmWave) frequency bands beyond 24 GHz, which provides tens of times more bandwidth~\cite{shafi20175g}.
However, to meet the unprecedentedly higher demand of nearly 50 exabyte-per-month data traffic anticipated for the future 6G and beyond wireless systems~\cite{traffic}, two main approaches may play a pivotal role to further enhance the wireless system capacity, namely,  massive multiple-input multiple-output (MIMO) and visible light communication (VLC), among others. 


On the one hand, thanks to the smaller wavelengths at future higher operating frequency, the base stations (BSs) and user devices can be equipped with more (tens or even hundreds of) antenna elements, thereby providing rich spatial degrees of freedom (DoFs) to enhance the wireless communication performance~\cite{tse2005fundamentals,cho2010mimo}.
This can provide a large array gain to compensate for the more severe path loss over higher frequency, as well as high spatial multiplexing gains to boost the system capacity, while enabling highly-directional beams to mitigate multi-user interference~\cite{lu2014overview}.
The above benefits of massive MIMO have been validated experimentally over mmWave band~\cite{rappaport2013millimeter}.

On the other hand, the next-generation wireless systems are anticipated to cater to diverse bandwidth-intensive services, such as connected robotics and autonomous systems, e-health, and eXtended reality (XR), which require 1000$\times$ increase in data rate, as well as more diverse wireless functions including communications, sensing, localization, etc., as compared to today's 5G wireless system. 
These stringent requirements may drive the spectrum migration from the existing sub-6 GHz and mmWave spectrums towards higher terahertz (THz) spectrum and even the optical spectrum, where VLC (400 THz-800 THz) is one of the major technologies~\cite{chi2020visible}.
In particular, with an extremely broad 400 THz license-free frequency band, VLC is promising to support explosive data growth. 
Furthermore, thanks to the low-cost VLC devices such as light-emitting diode (LED) and photodetector (PD), VLC can be densely deployed in the future wireless communication systems with low energy consumption and hardware costs~\cite{karunatilaka2015led}.
These advantages have spurred extensive research efforts in characterizing the performance limits of VLC. 
For example, the authors in~\cite{lapidoth2009capacity} characterized the capacity of a single-input single-output (SISO) VLC system, and the results were later extended to the multiple-input single-output (MISO)~\cite{8336902,8277933} and MIMO VLC systems~\cite{8006585}.
However, due to the strong spatial correlation over the optical spectrum, the multiplexing gain of MIMO VLC may suffer from significant performance loss~\cite{ying2015joint}, which 
limits the MIMO VLC capacity.
Although there have been several candidate approaches to tackle this challenge by using e.g., narrow field-of-view (FoV) PDs~\cite{ying2015joint}, imaging receiver~\cite{chen2016wide}, and/or angle-diversity receiver~\cite{nuwanpriya2015indoor}, they all need to modify existing VLC transmitter/receiver architecture and thus may face difficulty in practical implementation.

Recently, intelligent reflecting surface (IRS) has emerged as a promising technique to boost the wireless network performance in a cost-effective manner. 
Specifically, IRS is a digitally controlled metasurface comprising a massive number of passive reflecting elements, each able to tune the amplitude and/or phase of the impinging electromagnetic wave in real time~\cite{8466374}.
By integrating IRSs into wireless communication systems and jointly optimizing their reflections, the communication channels may be significantly reshaped to favor the signal transmission, without the need of modifying existing transmitter/receiver architectures~\cite{zhang2020capacity,9052904,9425508}.

Motivated by the above, in this paper, we propose to use an optical IRS (OIRS) to resolve the spatial correlation issue of VLC channels and thereby boost the MIMO VLC capacity.
It is worth mentioning that the OIRS has been studied in some prior works, on technological overview~\cite{jamali2021intelligent,Sun2021,9614037,aboagye2022ris}, channel modeling~\cite{2021intelligent,abdelhady2020visible,9681888}, performance optimization~\cite{sun_CL,sun2022joint,9543660,qian2021secure}, and so on~\cite{sun2022miso}. 
However, all of the above works only focused on the SISO/MISO VLC. 
In~\cite{wu2023configuring}, the authors studied a simplified OIRS-aided MIMO VLC system with a bunch of parallel SISO VLC channels.
As such, the general OIRS-aided MIMO VLC has not been studied in the existing works.
This paper fills in this gap and its main contributions are summarized as follows:
\begin{itemize}
	\item
	Firstly, as the electrical size of OIRS far exceeds that of conventional IRS working in radio-frequency (RF) band (e.g., sub-6G and mmWave communications), each OIRS reflecting element should be considered as an infinite large specular reflector, which results in the extremely near-field condition in VLC.
	Based on this, we develop a new OIRS channel model and reveal an interesting ``no crosstalk'' property of it, where the OIRS reflecting elements can be respectively configured to align with a specific pair of transmitter and receiver antennas, yet without causing any crosstalk to each other in their reflected links.
	
	\item
	Next, based on the developed OIRS channel model, we derive the tight lower and upper bounds on the OIRS-assisted MIMO VLC capacities under different types of emission power constraints at the transmitter, respectively, based on which the asymptotic capacity expressions are obtained in closed form in the high signal-to-noise-ratio (SNR) regime.
	Accordingly, we formulate a capacity maximization problem for the OIRS-assisted MIMO VLC, by jointly optimizing the OIRS's element alignment with the transmitter/receiver antenna pairs and the transmitter's emission power. 
	
	\item
	To solve the capacity maximization problem, we first show that it can be decoupled into two separate subproblems with respect to (w.r.t.) the OIRS element alignment and the transmitter emission power, respectively. 
	To solve the first subproblem, we propose two algorithms, referred to as location-aided interior-point (LIP) algorithm and log-det-based alternating optimization (LDAO) algorithm, which can balance the performance versus complexity trade-off.
	While for the second subproblem, we derive the optimal transmitter emission power in closed-form.
	Finally, numerical results are provided to validate the performance of the proposed algorithms in boosting the OIRS-assisted MIMO VLC capacity as compared to other baseline schemes.
\end{itemize}

The remainder of this paper is organized as follows. 
Section~\ref{Sec:Channel} and Section~\ref{Sec:MIMO VLC} present the channel model and capacity results of the OIRS-assisted MIMO VLC, respectively.
Section~\ref{Sec:Proposed} presents the formulated capacity maximization problems and the proposed algorithms to solve them.
Numerical results are presented in Section~\ref{Sec:Num} to evaluate the performance of the proposed algorithms.
Finally, Section~\ref{Sec:Conclude} concludes this paper.

\textit{Notations:}
In this paper, scalars are denoted by normal letters (e.g., $m$ and $M$), while vectors and matrices are denoted by small and large boldface letters (e.g., $\boldsymbol{m}$ and $\boldsymbol{M}$), respectively.
Moreover, the matrix transpose, vectorization, Frobenius norm, differential operator, and statistical expectation are denoted by $(\cdot)^T$, $\text{vec}(\cdot)$, $\left\|\cdot\right\|_F$, $\partial(\cdot)$ (or $\nabla(\cdot)$), and $\mathbb{E} [\cdot]$, respectively.
The Hadamard product, floor operator, and modulus operator are denoted by $\odot$, $\lfloor \cdot \rfloor$, and $\text{mod}(\cdot)$, respectively. 
The set $\mathbb{R}$ denotes the real number domain, while $\mathbb{R}^{+}$ denotes the real and non-negative number domain.
The $N \times N$ identity matrix, all-zero matrix, and all-one matrix are denoted as $\boldsymbol{I}_N$, $\boldsymbol{0}$, and $\boldsymbol{1}$, respectively.
Furthermore, diag($[\boldsymbol{m}_1, \cdots, \boldsymbol{m}_n]$) represents a block diagonal matrix, where the $i$th main diagonal vector is given by $\boldsymbol{m}_i$, and $[\boldsymbol{M}]_{i,j}$ or $m_{i,j}$ represents the element at the $i$th row and $j$th column of the matrix $\boldsymbol{M}$.
Finally, $\preceq$ means that each element at the left-hand side is smaller than or equal to its counterpart at the right-hand side.


\vspace{-0.2cm}
\section{System Model}
\vspace{-0.1cm}
\label{Sec:Channel}
As shown in Fig. 1, we consider an OIRS-assisted MIMO VLC system, where an $N$-element OIRS is deployed to assist in the optical communications from an $N_t$-LED transmitter to an $N_r$-PD receiver.
Due to the negligble non-line-of-sight (NLoS) channels in VLC, we assume line-of-sight (LoS) channels among the transmitter, receiver, and the OIRS. 
However, different from conventional RF-based wireless communications, the OIRS-assisted system has distinct properties in electrical size, signal model, power scaling law, and so on~\cite{jamali2021intelligent}.
As such, a new OIRS channel model is needed to account for these specific features. 
To this end, we first present three typical propagation conditions in RF or VLC, namely, far-field, near-field, and extremely near-field conditions, as detailed below.

\begin{figure}[t]
	\centering
	\includegraphics[width=0.6\textwidth]{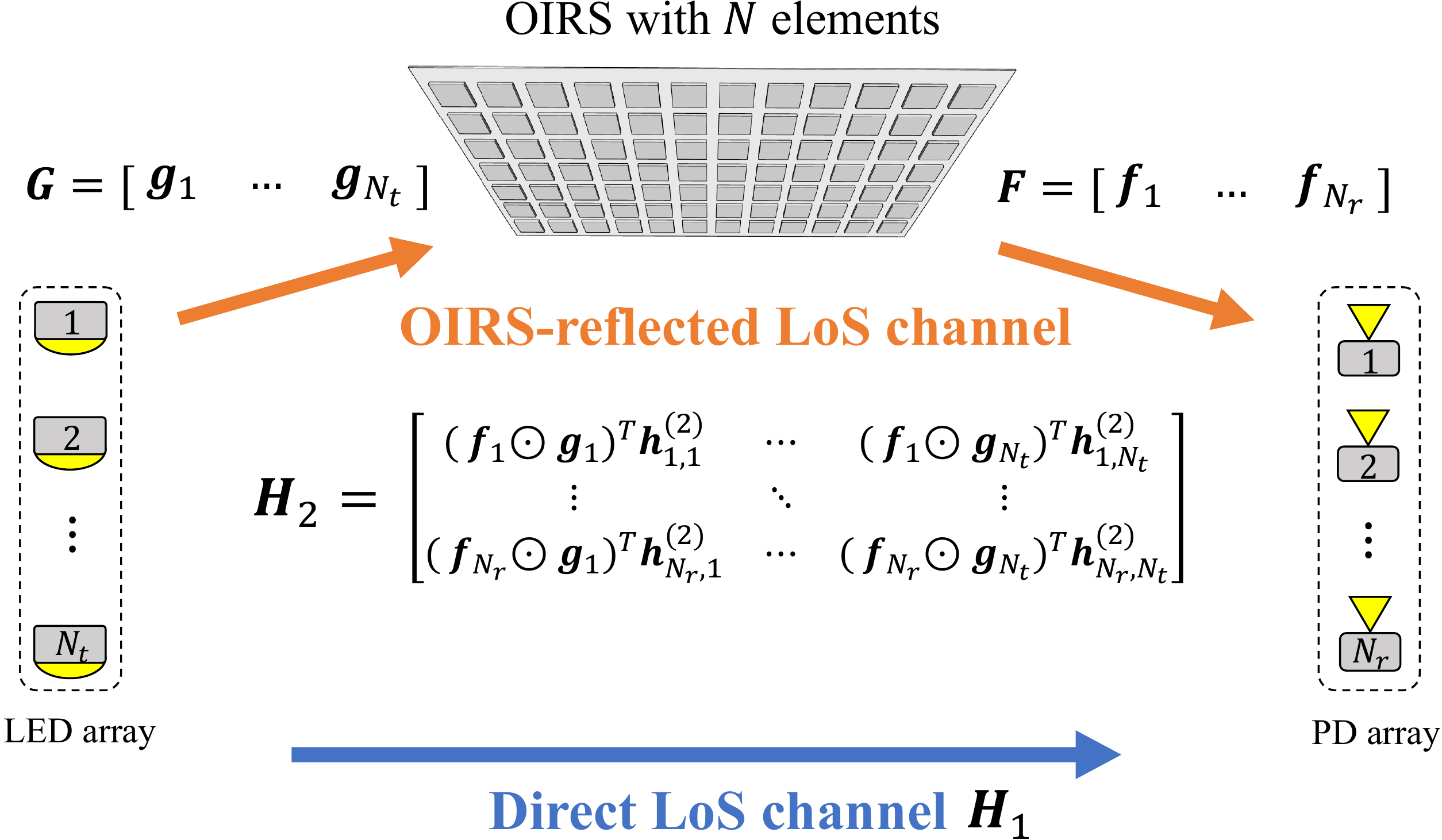}
	\caption{Channel model of the OIRS-assisted MIMO VLC.}
	\label{Fig:System}
	\vspace{-1.0cm}
\end{figure}

\vspace{-0.4cm}
\subsection{Comparison among Different Propagation Conditions}
\label{subsec:extremely}
\textbf{1) Far-field}:
The Rayleigh distance is often used as the demarcation between the near field and far field, which is given by
\begin{equation}
	\setlength\abovedisplayskip{5pt}
	\label{rayleigh}
	d_0 = \frac{2L^2f_c}{c},
	\setlength\belowdisplayskip{5pt}
\end{equation}
where $c$, $f_c$, and $L$ denote the light speed, the operating frequency, and the array aperture, respectively. 
When the distance between the transmitter/receiver and the IRS far exceeds~(\ref{rayleigh}), the far-field propagation occurs, under which the electromagnetic wave can be locally modeled as a planar wave at the IRS~\cite{sarieddeen2020next}.
Furthermore, the IRS-reflected path loss under the far-field propagation follows a ``multiplicative'' model, where it is inversely proportional to the square of the product of the transmitter-IRS distance and the IRS-receiver distance~\cite{wu2019towards}. 

\begin{figure*}[t]
	\centering
	\setlength{\abovecaptionskip}{0.3cm}
	\includegraphics[width=0.65\textwidth]{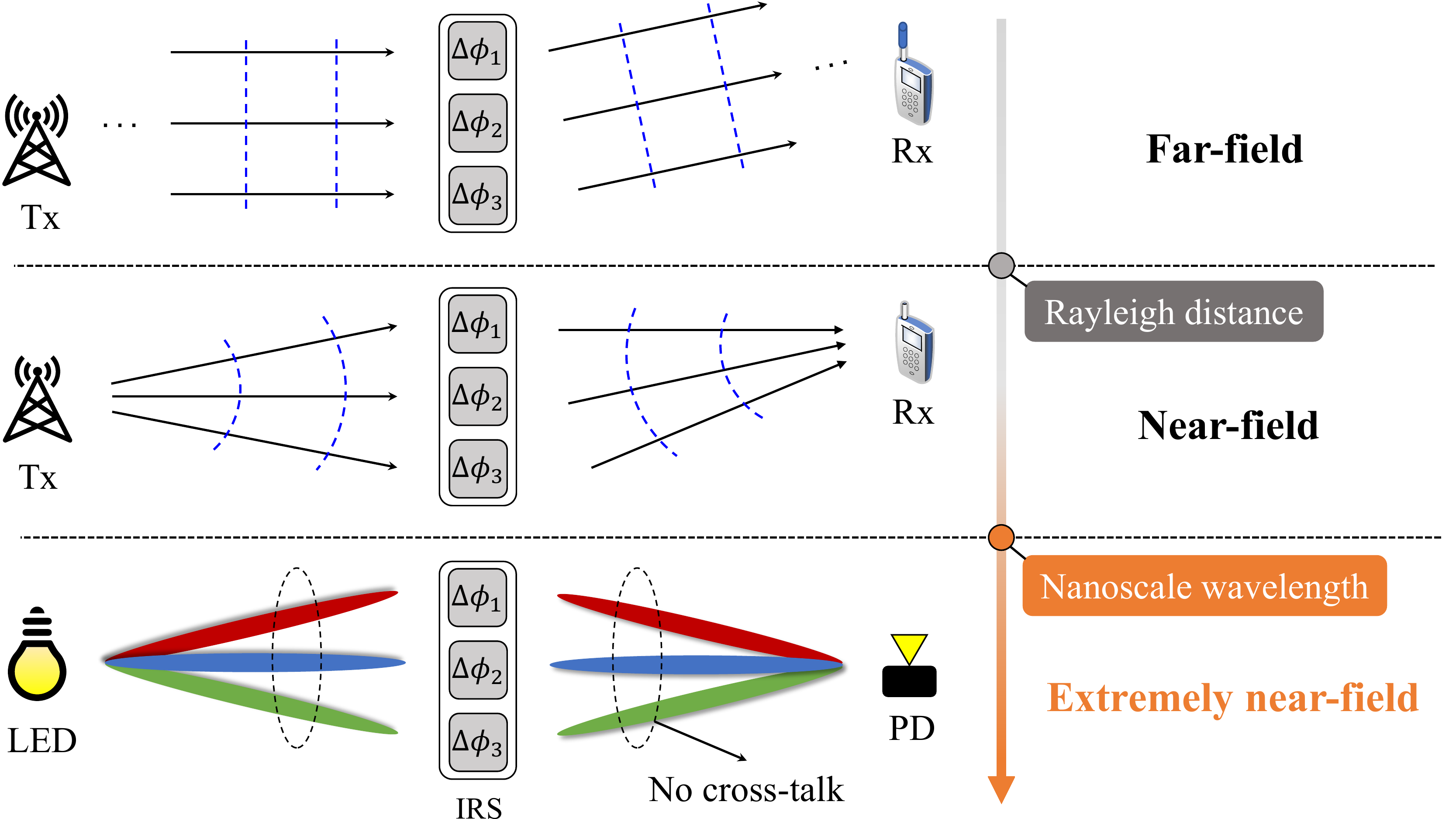}
	\caption{Illustration of the far-field, near-field, and extremely near-field propagation conditions.}
	\label{Fig:Near}
	\vspace{-0.9cm}
\end{figure*}

\textbf{2) Near-field}:
It is noted that the Rayleigh distance monotonically increases with $f_c$ and $L$. 
As such, when the operating frequency $f_c$ is sufficiently high, the near-field propagation will more frequently occur than the far-field one, e.g., in mmWave and THz communications.
The main difference between the far- and near-field propagation conditions lies in the phase differences among the IRS reflecting elements, which become nonlinear under the latter condition.
In this case, the incoming/outgoing signal at the IRS should be modeled as the spherical (instead of planar) wave~\cite{cui2022near}, which thus can realize the functionality of both beamforming and beamfocusing.
Moreover, the IRS-reflected path loss under the near-field propagation becomes closer to an ``additive'' model~\cite{tang2021wireless}, where it is inversely proportional to the square of the sum of the transmitter-IRS distance and the IRS-receiver distance.

\textbf{3) Extremely near-field}:
In optical spectrum, the operating wavelength $\lambda_c = c/f_c$ is generally from hundreds of nanometers to thousands of nanometers, which is negligible compared to the aperture of the OIRS reflecting element, thus giving rise to the so-called extremely near-field propagation condition. 
Different from the wave models in RF, the signal in VLC is modeled by the Lambertian wave, under which each reflecting element of the OIRS performs similarly as an infinitely large reflector~\cite{abdelhady2020visible}.
Hence, the directions of the incoming and outgoing electromagnetic waves at the OIRS can be approximately determined based on the geometric optics following the generalized Snell's law~\cite{2021intelligent}, with negligible energy leakage towards other undesired directions~\cite{jamali2021intelligent}.
Thanks to this peculiar property, if an OIRS reflecting element is configured to align with a specified pair of transmitter and receiver antennas, it will not interfere with the reflected paths by any other OIRS elements (each aligning with either this or another pair of transmitter/receiver antenna), which is referred to as the ``no cross-talk'' property of the OIRS-reflected channels, as illustrated at the bottom part of Fig.~\ref{Fig:Near}.
It follows that for any given pair of transmitter and receiver antennas, the reflection amplitude of each OIRS reflecting element can be equivalently viewed as a binary variable of $0$ or $1$, depending on whether they are aligned.
In summary, we compare the properties of the above propagation conditions in Table~\ref{Extremely-near-field}.

\newcommand{\tabincell}[2]{\begin{tabular}{@{}#1@{}}#2\end{tabular}}
\renewcommand\arraystretch{1.5}
\begin{table*}[t]
	\centering
	\small
	\caption{Comparison of far-field, near-field, and extremely near-field communications.}
	\vspace{-0.4cm}
	\label{Extremely-near-field}
	\begin{tabular}{| l | c | c | c |}
		\hline
		& \textbf{Extremely near-field} & \textbf{Near-field} & \textbf{Far-field} \\
		\hline
		Operating frequency & Optical range & mmWave and THz ranges & Low-frequency RF range \\
		\hline
		Wave model & Lambertian wave & Spherical wave & Planar wave \\
		\hline
		Path loss ($PL$)\footnotemark[1]  & $PL \propto (d_1 + d_2)^2$  & $PL \overset{\approx}{\propto} (d_1 + d_2)^2$ & $PL \propto d_1^2d_2^2$ \\
		\hline
		Channel model & ``Additive'' model & Nearly  ``additive'' model & ``Multiplicative'' model \\
		\hline
	\end{tabular}
	\vspace{-0.8cm}
\end{table*}
\footnotetext[1]{$d_1$ and $d_2$ denote the transmitter-IRS and IRS-receiver distances, respectively.}

\vspace{-0.3cm}
\subsection{Channel Model of OIRS-Assisted MIMO VLC}
\vspace{-0.2cm}
\label{Subsec:Gain}
In this subsection, we develop the channel model of the OIRS-assisted MIMO VLC system based on the extremely near-field condition.
Note that the VLC channels are approximately frequency-flat since hardware devices can hardly distinguish the nanosecond delay of the multipath.
Accordingly, we denote by $\boldsymbol{H}_1$ and $\boldsymbol{H}_2  \in \mathbb{R}_{N_r\times N_t}^+$ as the direct and OIRS-reflected LoS MIMO channels from the LED to the PD, respectively.
Moreover, let $\mathcal{T}\triangleq\{1,\cdots,N_t\}$, $\mathcal{R}\triangleq\{1,\cdots,N_r\}$, and $\mathcal{N}\triangleq\{1,\cdots,N\}$ denote the index sets of the LEDs, PDs, and OIRS reflecting elements, respectively.
The detailed expressions of $\boldsymbol{H}_1$ and $\boldsymbol{H}_2$ are provided as follows.


\textbf{Direct LoS channel}:
According to the Lambertian radiant formula, the direct LoS channel from the $i$th LED to the $j$th PD is given by~\cite{obeed2019optimizing}
\begin{equation}
	\label{LoS_VLC_Link}
		\setlength\abovedisplayskip{7pt}
	h_{j,i}^{\left(1\right)} =
	\left\{
	\begin{aligned}
		&\frac{q^2(m+1)A_{PD}}{2\pi D_{j,i}^2 \sin^2(\Phi_0)}g_{of}\cos^m(\Theta_{j,i})\cos(\Phi_{j,i}), & &\text{if}\ 0 \leq \Phi_{j,i} \leq \Phi_0,   \\
		&0, & &\text{otherwise},
	\end{aligned}
	\right.
		\setlength\belowdisplayskip{7pt}
\end{equation}
where $q$, $m$, $A_{PD}$, $g_{of}$, and $\Phi_0$ denote the refractive index, the Lambertian index, the area of PD, the optical filter gain, and the semi-angle of FoV, respectively;
$D_{j,i}$ denotes the distance from the $i$th LED to the $j$th PD, $\Theta_{j,i}$ denotes the angle of irradiance at the LED array, and $\Phi_{j,i}$ denotes the angle of incidence at the PD array.
Based on~(\ref{LoS_VLC_Link}), the direct LoS MIMO channel from the LED array to the PD array can be expressed as $\boldsymbol{H}_1 \triangleq [\boldsymbol{h}_1^{(1)}, \cdots, \boldsymbol{h}_{N_t}^{(1)}]$, with $\boldsymbol{h}_{i}^{(1)} \triangleq [h_{1,i}^{(1)},\cdots,h_{N_r,i}^{(1)}]^T$.

\textbf{OIRS-reflected LoS channel}:
As previously discussed in Section~\ref{subsec:extremely}, under the extremely near-field condition, the OIRS-reflected LoS channel presents approximately no cross-talk among its reflecting elements. 
To depict this effect, we define two binary matrices $\boldsymbol{G} \triangleq [ \boldsymbol{g}_1, \cdots, \boldsymbol{g}_{N_t} ] \in\{0,1\}_{N\times N_t}$ and $\boldsymbol{F} \triangleq [ \boldsymbol{f}_1, \cdots, \boldsymbol{f}_{N_r} ] \in\{0,1\}_{N\times N_r}$ to indicate the transmitter/receiver antenna pair aligned with each OIRS reflecting element.
More specifically, if $g_{n, i} = f_{n, j} = 1$, the $n$th OIRS reflecting element is aligned with the $i$th LED and the $j$th PD. 
Then, based on the point source assumption~\cite{obeed2019optimizing,abdelhady2020visible}, their formed cascaded channel can be approximately expressed as
\begin{equation}
	\label{NonLoS_VLC_Link}
	\setlength\abovedisplayskip{5pt}
	h_{j,n,i}^{\left(2\right)} =
	\left\{
	\begin{aligned}
		&\frac{\gamma q^2(m+1)A_{PD}}{2\pi (d_{n,i}^{1}+d_{j,n}^{2})^2 \sin^2(\Phi_0)}g_{of}\cos^m(\theta_{n,i})\cos(\phi_{j,n}), & &\text{if}\ 0 \leq \phi_{j,n} \leq \Phi_0, \\
		&0, & &\text{otherwise},
	\end{aligned}
	\right.
	\setlength\belowdisplayskip{5pt}
\end{equation}
where $\gamma$, $d_{n,i}^{1}$, and $d_{n,j}^{2}$ denote the reflectivity of the OIRS reflecting element, the distance from the $i$th LED to the $n$th OIRS reflecting element, and the distance from the $n$th OIRS reflecting to the $j$th PD, respectively; and the angle of irradiance at the $i$th LED and the angle of incidence at the $j$th PD are denoted as $\theta_{n,i}$ and $\phi_{j,n}$, respectively.
Accordingly, we define $\boldsymbol{h}_{j,i}^{(2)} \triangleq [h_{j,1,i}^{(2)}, \cdots, h_{j,N,i}^{(2)}]^T$ as the OIRS-reflected LoS channel from the $i$th LED to $j$th PD.
Evidently, only one of the $N$ entries in $\boldsymbol{h}_{j,i}^{(2)}$ is non-zero.

As each OIRS reflecting element can only be tuned to align with one pair of transmitter and receiver antennas, we have
\begin{align}
	\setlength\abovedisplayskip{3pt}
	\sum_{i=1}^{N_t}g_{n,i} \leq 1,\quad \sum_{j=1}^{N_r}f_{n, j}\leq1, \quad \forall\ \! n \in \mathcal{N}. \label{Eq:Row summation}
	\setlength\belowdisplayskip{3pt}
\end{align}
It follows from the above that the OIRS-reflected LoS channel from the $i$th LED to the $j$th PD depends on both the reflected channel gain and the alignment at the OIRS (as specified by $\boldsymbol{f}_j$ and $\boldsymbol{g}_i$), which can be written as~\cite{sun2022mimo}
\vspace{-0.3cm}
\begin{align}
	\setlength\abovedisplayskip{3pt}
	\label{Eq:H2_generation}
	\left[\boldsymbol{H}_2\right]_{j,i} = \left(\boldsymbol{f}_{j} \odot \boldsymbol{g}_{i}\right)^T\boldsymbol{h}_{j,i}^{(2)}.
	\setlength\belowdisplayskip{3pt}
\end{align}
It is noted from~(\ref{Eq:H2_generation}) that the $n$th OIRS reflecting element contributes to $[\boldsymbol{H}_2]_{j,i}$ if and only if $g_{n, i} = f_{n, j} = 1$, under which we have $[\boldsymbol{H}_2]_{j,i} = h_{j,n,i}^{(2)}$.

\vspace{-0.3cm}
\section{OIRS-Assisted MIMO VLC Capacity Characterization}
\label{Sec:MIMO VLC}
In this section, based on the developed OIRS channel model, we derive the exact and asymptotic (in the high-SNR regime) capacities of the OIRS-assisted MIMO VLC under three types of practical power constraints, where we assume different emission power at LEDs, in contrast to existing works assuming identical emission power~\cite{8006585}.
Specifically, the received signal at the PD array can be expressed as 
\begin{equation}\\
	\setlength\abovedisplayskip{3pt}
	\label{Eq:signal model}
	\boldsymbol{y} = \boldsymbol{H}\boldsymbol{x} + \boldsymbol{z},
	\setlength\belowdisplayskip{3pt}
\end{equation}
where $\boldsymbol{x} \triangleq [x_1,\cdots,x_{N_t}]^T \in \mathbb{R}_{N_t\times 1}^+$ denotes the optical intensity emitted from the LED array (transmitter), $\boldsymbol{y} \triangleq [y_1,\cdots,y_{N_r}]^T \in \mathbb{R}_{N_r\times 1}^+$ denotes the received signal at the PD array (receiver), and $\boldsymbol{z}\sim N(\boldsymbol{0},\boldsymbol{K})$ denotes the thermal noise with $\boldsymbol{K}$ denoting its covariance matrix.
The MIMO channel is given by the superposition of the direct and OIRS-reflected LoS components, i.e.,
\begin{equation}
	\setlength\abovedisplayskip{3pt}
	\label{Eq:channel}
	\boldsymbol{H} = \boldsymbol{H}_1 + \boldsymbol{H}_2.
	\setlength\belowdisplayskip{3pt}
\end{equation}

\footnotetext[2]{Note that different from the RF communication, VLC signal is modulated on the amplitude of the emitting light. In this regard, we interchangeably use the term ``intensity'' and ``power'' in this paper.}

In VLC, the input $\boldsymbol{x}$ is generally subjected to two typical optical intensity constraints\footnotemark[2], namely, peak optical intensity and average optical intensity constraints~\cite{lapidoth2009capacity}.
Firstly, for the peak optical intensity constraint, let $\boldsymbol{A} \triangleq [A_1,\cdots,A_{N_t}]^T  \in \mathbb{R}_{N_t\times 1}^+$ be the peak intensity vector at the LED transmitter, where $A_i$ denotes the peak intensity of the $i$th LED, and let $\text{Pr}(\cdot)$ denote the probability measurement. 
Then, the optical intensity per LED should satisfy
\begin{equation}
	\setlength\abovedisplayskip{3pt}
	\label{Eq:peak constriant}
	\text{Pr}\left( x_i > A_i \right) = 0, \quad \forall i\in \mathcal{T}.
	\setlength\belowdisplayskip{3pt}
\end{equation}
Secondly, for the average optical intensity constraint, let $\boldsymbol{E} \triangleq [E_1,\cdots,E_{N_t}]^T  \in \mathbb{R}_{N_t\times 1}^+$ be the average intensity vector at the LED transmitter, where $E_i$ denotes the average intensity of the $i$th LED.
Then, the optical intensity per LED should satisfy 
\begin{equation}
	\setlength\abovedisplayskip{3pt}
	\label{Eq:average constriant}
	\mathbb{E}\left( x_i \right) < E_i.
	\setlength\belowdisplayskip{3pt}
\end{equation}
Meanwhile, the total average optical intensity of the LED array is constrained by
\begin{equation}
	\setlength\abovedisplayskip{3pt}
	\label{Eq:total_average constriant}
	\sum_{i = 1}^{N_t}\mathbb{E}\left( x_i \right) < E,
	\setlength\abovedisplayskip{3pt}
\end{equation}
\vspace{-0.8cm}

\noindent
where $E \triangleq \sum_{i=1}^{N_t} E_i$ denotes the total average intensity. 
Given these two types of optical intensity constraints, we further define the ratio of the sum average intensity to the sum peak intensity as $\alpha$, which is expressed as
\begin{equation}
	\setlength\abovedisplayskip{3pt}
	\label{Eq:alpha}
	\alpha = \frac{E}{\sum_{i = 1}^{N_t} A_i}.
	\setlength\belowdisplayskip{3pt}
\end{equation}
Note that since the total average intensity cannot exceed the total peak intensity, we have $\alpha \in (0,1]$. 

In particular, the exact capacity of the MIMO VLC system is given by~\cite{lapidoth2009capacity,8006585}
\begin{equation}
	\setlength\abovedisplayskip{3pt}
	\label{Eq:capacity}
	C(\boldsymbol{A}, E) = \mathop{\text{sup}}\limits_{Q\left(\boldsymbol{x}\right)}\textit{I}\left( Q\left(\boldsymbol{x}\right); W\left(\boldsymbol{y}|\boldsymbol{x}\right) \right),
	\setlength\belowdisplayskip{3pt}
\end{equation}
where ``sup'' denotes the supremum operator, $\textit{I}(\boldsymbol{x};\boldsymbol{y})$ denotes the mutual information between the input and output signals, $Q(\boldsymbol{x})$ denotes the probability measure of the input $\boldsymbol{x}$, and $W(\boldsymbol{y}|\boldsymbol{x})$ denotes the transition probability measure from $\boldsymbol{x}$ to $\boldsymbol{y}$.
It has also been shown in~\cite{8006585} that the MIMO VLC capacity varies with $\alpha$.
However, the capacity-achieving distribution of the input $\boldsymbol{x}$ in VLC is in a discrete form~\cite{obeed2019optimizing}, thus making it difficult to obtain a closed-form expression of~(\ref{Eq:capacity}). 
In the following, we aim to derive its tight upper and lower bounds under different types of intensity constraints, which help to obtain the asymptotic capacity expression in the high-SNR regime. 

Next, we consider the following three cases of intensity constraints. 
In Cases \uppercase\expandafter{\romannumeral1} and \uppercase\expandafter{\romannumeral2}, we consider both the peak intensity and average intensity constraints with $\alpha\in (0,1/2)$ and $\alpha\in [1/2,1]$, respectively. 
While in Case \uppercase\expandafter{\romannumeral3}, only the average intensity constraint is considered with $\alpha\to0$. 
The corresponding capacity results in the above three cases are presented in the following propositions.

\begin{proposition}
	\label{pro1}
	In Case \uppercase\expandafter{\romannumeral1}, the capacity of the OIRS-assisted MIMO VLC in the high-SNR regime, i.e., $A_i \to \infty, \forall i$, is given by
	\begin{align}
	\setlength\abovedisplayskip{3pt}
	\label{Eq:MIMOcase1_capacity}
	C\left( \boldsymbol{G}, \boldsymbol{F}, \boldsymbol{A} \right) = \sum_{i = 1}^{N_t}\log A_i + \frac{1}{2}\log\det\left(\boldsymbol{H}^T\boldsymbol{K}^{-1}\boldsymbol{H}\right) + \chi\left( \alpha \right),\quad \alpha \in (0,\frac{1}{2}),
	\setlength\belowdisplayskip{3pt}
	\end{align}
	where 
	\begin{align}
		\setlength\abovedisplayskip{3pt}
		\label{Eq:chi_1}
		\chi\left( \alpha \right) \triangleq - N_t \left( \frac{\log2\pi e}{2} + \log\left(1-\alpha\mu^*\right) + \mu^*\left( 1-\alpha \right) \right),
		\setlength\belowdisplayskip{3pt}
	\end{align}
	and $\mu^*$ is the root of the equation
	\begin{equation}
		\label{Eq:MIMOcase1_mu}
		\alpha = \frac{1}{\mu^*} - \frac{e^{-\mu^*}}{1-e^{-\mu^*}}.
	\end{equation}
\end{proposition}

\begin{proof}
	Proposition~\ref{pro1} can be proved by first deriving the tight lower and upper bounds on~(\ref{Eq:capacity}). 
	Their expressions and the detailed proof are given in Appendix~\ref{app:A}.
\end{proof}
	
Note that the asymptotic capacity in~(\ref{Eq:MIMOcase1_capacity}) depends on the peak optical intensity $\boldsymbol{A}$ and OIRS element alignments $\boldsymbol{G}$ and $\boldsymbol{F}$ (as subsumed in $\boldsymbol{H}$ via~(\ref{Eq:H2_generation})). 
In Case \uppercase\expandafter{\romannumeral2}, we first present the following lemma to simplify the exact capacity in~(\ref{Eq:capacity}).

\begin{lemma}
	\label{Lemma:thres1/2}
	When $\alpha > 1/2$, the capacity of the OIRS-assisted MIMO VLC in~(\ref{Eq:capacity}) can be simplified as
	\begin{equation}
		\setlength\abovedisplayskip{3pt}
		\label{Eq:MIMO_thres1/2_capacity}
		C(\boldsymbol{A}) = C\left(\boldsymbol{A}, \frac{\sum_{i = 1}^{N_t} A_i}{2}\right),
		\setlength\belowdisplayskip{3pt}
	\end{equation}
	with the capacity-achieving distribution $Q^*(\boldsymbol{x})$ satisfying
	\begin{equation}
		\setlength\abovedisplayskip{1pt}
		\label{Eq:1/2bestQ}
		\mathbb{E}_{Q^*}\left[ \boldsymbol{x} \right] = \frac{1}{2}\boldsymbol{A}.
		\setlength\belowdisplayskip{1pt}
	\end{equation}
\end{lemma}
\begin{proof}
	See Appendix~\ref{app:B}.
\end{proof}

Note that Lemma~\ref{Lemma:thres1/2} generalizes~\cite[Proposition 1]{8006585} as a special case, where the peak intensity at all transmitter antennas are assumed to be identical, i.e., $A_i=A,\ \forall i \in\mathcal{T}$. 
Based on Lemma~\ref{Lemma:thres1/2}, we can derive the asymptotic capacity expression in Case \uppercase\expandafter{\romannumeral2} as follows.

\begin{proposition}
	\label{pro2}
	In Case \uppercase\expandafter{\romannumeral2}, the capacity of the OIRS-assisted MIMO VLC in the high-SNR regime, i.e., $A_i \to \infty, \forall i$, is given by
	\begin{align}
	\setlength\abovedisplayskip{3pt}
	\label{Eq:MIMOcase2_capacity}
	C\left( \boldsymbol{G}, \boldsymbol{F}, \boldsymbol{A} \right) = \sum_{i = 1}^{N_t}\log A_i + \frac{1}{2}\log\det\left(\boldsymbol{H}^T\boldsymbol{K}^{-1}\boldsymbol{H}\right) + \chi\left( \alpha \right), \quad \alpha \in [\frac{1}{2},1],
	\setlength\belowdisplayskip{3pt}
	\end{align}
	where
	\begin{align}
		\setlength\abovedisplayskip{3pt}
		\label{Eq:chi_2}
		\chi\left( \alpha \right) \triangleq - \frac{N_t \log2\pi e}{2}.
		\setlength\belowdisplayskip{3pt}
	\end{align}
\end{proposition}

\begin{proof}
	Proposition~\ref{pro2} can be proved by first deriving the lower and upper bounds on~(\ref{Eq:MIMO_thres1/2_capacity}). 
	Their expressions and the detailed proof are given in Appendix~\ref{app:C}.
\end{proof}

Finally, in Case \uppercase\expandafter{\romannumeral3}, by applying the similar mathematical manipulations as in Propositions~\ref{pro1} and~\ref{pro2}, we can directly obtain its asymptotic capacity expression as follows.

\begin{proposition}
	In Case \uppercase\expandafter{\romannumeral3}, the capacity of the OIRS-assisted MIMO VLC in the high-SNR regime, i.e., $E_i \to \infty, \forall i$, is given by
	\begin{align}
	\setlength\abovedisplayskip{3pt}
	\label{Eq:MIMOcase3_capacity}
	C\left( \boldsymbol{G}, \boldsymbol{F}, \boldsymbol{E} \right) = \sum_{i = 1}^{N_t}\log E_i + \frac{1}{2}\log\det\left(\boldsymbol{H}^T\boldsymbol{K}^{-1}\boldsymbol{H}\right) - \frac{N_t}{2}\log\frac{2\pi N_t^2}{e}, \quad \alpha\to0.
	\setlength\belowdisplayskip{3pt}
	\end{align}
\end{proposition}


Based on the above, we summarize the asymptotic capacity under the above three cases in Table~\ref{Tab:capacity}. 
It is observed that the capacity results in these three cases are all given by the sum of a specific constant term and a common term below,
\begin{equation}
	\label{Eq:objective}
	f\left( \boldsymbol{G}, \boldsymbol{F}, \boldsymbol{X} \right) \triangleq \frac{1}{2}\log\det\left(\boldsymbol{H}^T\boldsymbol{K}^{-1}\boldsymbol{H}\right) + \sum_{i = 1}^{N_t}\log X_i, \quad \boldsymbol{X} \in \{\boldsymbol{A}, \boldsymbol{E}\}.
\end{equation}
Therefore, to maximize the asymptotic capacity of the OIRS-assisted MIMO VLC in the high-SNR regime, it is equivalent to maximize the common term $f\left( \boldsymbol{G}, \boldsymbol{F}, \boldsymbol{X} \right)$ in~(\ref{Eq:objective}) in all considered cases, by jointly optimizing the OIRS's element alignment $\boldsymbol{G}$ and $\boldsymbol{F}$, as well as the intensity vector $\boldsymbol{A}$ or $\boldsymbol{E}$, as will be studied in the next section.

\renewcommand\arraystretch{1.5}
\begin{table*}[t]
	\centering
	\small
	\caption{Summary of the OIRS-assisted MIMO VLC capacities in the high-SNR regime.}
	\vspace{-0.3cm}
	\begin{tabular}{| l | c | c | c |}
		\hline
		Case & \uppercase\expandafter{\romannumeral1} & \uppercase\expandafter{\romannumeral2} & \uppercase\expandafter{\romannumeral3} \\
		\hline
		Intensity constraint & \multicolumn{2}{c|}{Peak intensity and average intensity constraints} & Only average intensity constraint \\
		\hline
		Range of $\alpha$ & $\alpha\in (0,1/2)$ & $\alpha\in [1/2,1]$ & $\alpha\to0$ \\
		\hline
		Capacity results & \tabincell{l}{$f\left( \boldsymbol{G}, \boldsymbol{F}, \boldsymbol{A} \right) + \chi\left( \alpha \right)$} & \tabincell{l}{$f\left( \boldsymbol{G}, \boldsymbol{F}, \boldsymbol{A} \right) + \chi\left( \alpha \right)$} & \tabincell{l}{$f\left( \boldsymbol{G}, \boldsymbol{F}, \boldsymbol{E} \right) - \frac{N_t}{2}\log\frac{2\pi N_t^2}{e}$} \\
		\hline
		Shared common term & \multicolumn{3}{c|}{$\frac{1}{2}\log\det\left(\boldsymbol{H}^T\boldsymbol{K}^{-1}\boldsymbol{H}\right) + \sum_{i = 1}^{N_t}\log X_i $ }\\
		\hline
	\end{tabular}
	\label{Tab:capacity}
\end{table*}
\setlength{\textfloatsep}{0.5cm}

\vspace{-0.4cm}
\section{OIRS-Assisted MIMO VLC Capacity Maximization}
\vspace{-0.1cm}
\label{Sec:Proposed}
In this section, we formulate the capacity maximization problem for the OIRS-assisted MIMO VLC system to maximize~(\ref{Eq:objective}), and solve it accordingly.

\vspace{-0.4cm}
\subsection{Problem Formulation}
\vspace{-0.1cm}
\label{Subsec:MIMO formulation}
For convenience, we only consider Cases \uppercase\expandafter{\romannumeral1} and \uppercase\expandafter{\romannumeral2} in this section, with $\boldsymbol{X}=\boldsymbol{A}$ in~(\ref{Eq:objective}), while the proposed algorithms are also applicable to Case \uppercase\expandafter{\romannumeral3}, by replacing $\boldsymbol{A}$ with $\boldsymbol{E}$. 
In view of the illumination requirements in VLC, we consider different peak power values for the LEDs at the transmitter, so as to cater to different human-centric activities. 
For example, reading books and using computers request different light intensities from 30 lx to 500 lx.
Let $A_{i,\max}$ denote the maximum emission power of the $i$th LED. 
Then, we have
\begin{equation}
	\label{Eq:MIMO_power_constraint}
	\boldsymbol{A} \preceq \boldsymbol{A}_{\max},
\end{equation}
where $ \boldsymbol{A}_{\max} \triangleq [A_{1,\max}, \cdots, A_{n,\max}] \in \mathbb{R}_{N_t\times 1}^+$ denotes the vector of the maximum emission power of the LED array.
Moreover, as the total emission power of the VLC system is generally constrained for saving energy, we further consider the following total power constraint,
\begin{equation}
	\setlength\abovedisplayskip{3pt}
	\label{Eq:MIMO_totalpower_constraint}
	\sum_{i = 1}^{N_t}A_i \leq A_{\text{total}},
	\setlength\belowdisplayskip{3pt}
\end{equation}
where $A_{\text{total}}$ denotes the total emission power of the $N_t$ LEDs.
Based on the above, the capacity maximization problem of the OIRS-assisted MIMO VLC can be formulated as
\begin{align}
	\setlength\abovedisplayskip{3pt}
	\text{(P0):}\ \max\limits_{\begin{subarray}{c}
			\boldsymbol{G}, \boldsymbol{F}, \boldsymbol{A}
	\end{subarray}}\ &\  f\left( \boldsymbol{G}, \boldsymbol{F}, \boldsymbol{A} \right) \label{P:objective}\\
	\text{s.t.}\ 
	& f_{n,j},\ g_{n,i}\in\{0,1\}, \quad \forall\ \! i \in \mathcal{T},\ j \in \mathcal{R}, \label{P:dis} \\
	& \sum_{i=1}^{N_t}g_{n,i} \leq 1, \quad \forall\ \! n \in \mathcal{N}, \label{P:g_row}\\
	& \sum_{j=1}^{N_r}f_{n, j} \leq 1, \quad \forall\ \! n \in \mathcal{N}, \label{P:f_row}\\
	& \boldsymbol{A} \preceq \boldsymbol{A}_{\max}, \label{P:power} \\
	& \sum_{i = 1}^{N_t}A_i \leq A_{\text{total}}. \label{P:total_power}
	\setlength\belowdisplayskip{3pt}
\end{align}

It is noted that~(\ref{Eq:objective}) is challenging to be optimally solved for the following reasons. 
Firstly, its objective function is non-concave in $\boldsymbol{G}$, $\boldsymbol{F}$, and $\boldsymbol{A}$. 
Secondly, it involves integer programming due to the constraints in~(\ref{P:dis}). 
Nonetheless, it is also noted that the OIRS element alignment ($\boldsymbol{G}$ and $\boldsymbol{F}$) and the transmitter emission power vector ($\boldsymbol{A}$) can be decoupled in the two terms $\log\det(\boldsymbol{H}^T\boldsymbol{K}^{-1}\boldsymbol{H})$ and $\sum_{i = 1}^{N_t}\log A_i$, respectively.
Therefore, (P0) can be decoupled into two subproblems and solved separately.
More specifically, the subproblem w.r.t. the OIRS element alignment can be expressed as
\vspace{-0.3cm}
\begin{align}
	\setlength\abovedisplayskip{3pt}
	\text{(P1):}\ \max\limits_{\begin{subarray}{c}
			\boldsymbol{G}, \boldsymbol{F}
	\end{subarray}}\ & f_1\left( \boldsymbol{G}, \boldsymbol{F} \right) \triangleq \log\det\left(\boldsymbol{H}^T\boldsymbol{K}^{-1}\boldsymbol{H}\right) \label{P1:objective}\\
	\text{s.t.}\ &~(\ref{P:dis})-(\ref{P:f_row}),
	\setlength\belowdisplayskip{3pt}
\end{align}
while the subproblem w.r.t. the transmitter emission power is given by
\vspace{-0.3cm}
\begin{align}
	\setlength\abovedisplayskip{3pt}
	\text{(P2):}\ \max\limits_{\begin{subarray}{c}
			\boldsymbol{A}
	\end{subarray}}\ & f_2\left( \boldsymbol{A} \right) \triangleq \sum_{i = 1}^{N_t}\log A_i \label{P2:objective}\\
	\text{s.t.}\ &~(\ref{P:power}),~(\ref{P:total_power}).
	\setlength\belowdisplayskip{3pt}
\end{align}
Next, we present our proposed algorithms for solving (P1) and (P2), respectively.


\vspace{-0.2cm}
\subsection{Proposed Solution to (P1)}
Note that (P1) is an integer programming problem, which lacks a universal polynomial-time solving algorithm in general. 
In the following, we propose two suboptimal algorithms, namely, LIP and LDAO, to solve (P1), which can strike a flexible balance between the performance versus complexity. 
Specifically, the LIP algorithm replaces the OIRS element alignment matrices, $\boldsymbol{G}$ and $\boldsymbol{F}$, with an auxilliary matrix $\boldsymbol{V}$ and searches the solution along the gradient direction based on the prior location information.
While the LDAO algorithm optimizes $\boldsymbol{G}$ and $\boldsymbol{F}$ in an alternate manner by exploiting the log-det heuristic (LDH)~\cite{fazel2003log}. 
Their detailed procedures are provided below.

\textbf{1) LIP algorithm}: 
Firstly, we introduce an auxillary matrix variable $\boldsymbol{V} \triangleq [ \boldsymbol{v}_1, \cdots, \boldsymbol{v}_{N_tN_r} ] \in\{0,1\}_{N\times N_tN_r}$ as~\cite{sun2022mimo}
\begin{equation}
	\setlength\abovedisplayskip{3pt}
	\label{Eq:v_vector}
	\boldsymbol{v}_{j + (i - 1)N_r} = \boldsymbol{f}_{j} \odot \boldsymbol{g}_{i}, \quad \forall\ \! i \in \mathcal{T},\ j \in \mathcal{R}.
	\setlength\belowdisplayskip{3pt}
\end{equation}
Note that the indices of the non-zero entries in~(\ref{Eq:v_vector}) (if any) indicate the indices of the OIRS reflecting elements aligning with the $i$-th LED and the $j$-th PD. 
Then, the end-to-end channel from the transmitter to the receiver can be rewritten as
\begin{equation}
	\setlength\abovedisplayskip{3pt}
	\label{Eq:MIMO_Vlinear}
	\text{vec}\left( \boldsymbol{H} \right) = \text{vec}\left( \boldsymbol{H}_1 \right) + \text{diag}\left( \boldsymbol{H}_c \right)^T\text{vec}\left( \boldsymbol{V} \right),
	\setlength\belowdisplayskip{3pt}
\end{equation}
where $\boldsymbol{H}_c \triangleq [ \boldsymbol{h}_{1,1}^{(2)}, \cdots, \boldsymbol{h}_{N_r,1}^{(2)}, \boldsymbol{h}_{1,2}^{(2)}, \cdots, \boldsymbol{h}_{N_r,N_t}^{(2)} ]$.
It is observed that~(\ref{Eq:MIMO_Vlinear}) is linear in $\boldsymbol{V}$.

Hence, (P1) is transformed into the following equivalent problem as
\vspace{-0.3cm}
\begin{align}
	\setlength\abovedisplayskip{3pt}
	\text{(P1-a):}\ \max\limits_{\begin{subarray}{c}
			\boldsymbol{V}
	\end{subarray}}\  &\log\det\left(\boldsymbol{H}^T\boldsymbol{K}^{-1}\boldsymbol{H}\right) \label{P1-a:objective}\\
	\text{s.t.}\ 
	& v_{n,p}\in\{0,1\}, \quad \forall\ \! n \in \mathcal{N},\ p \in \mathcal{P}, \label{P1-a:dis}\\
	& \sum_{p=1}^{N_tN_r}v_{n,p} \leq 1, \quad \forall\ \! n \in \mathcal{N}. \label{P1-a:g_row}
	\setlength\belowdisplayskip{3pt}
\end{align}
To tackle this problem, we relax $\boldsymbol{V}$ into continuous variables and solve this relaxed problem, and then reconstruct the optimized solution in discrete form via the minimum distance criterion.
Specifically, by relaxing the entries of $\boldsymbol{V}$ into continuous variables between $0$ and $1$, (P1-a) becomes
\begin{align}
	\setlength\abovedisplayskip{3pt}
	\text{(P1-a-$r$):}\ \max\limits_{\boldsymbol{V}}&\  \log\det\left(\boldsymbol{H}^T\boldsymbol{K}^{-1}\boldsymbol{H}\right) \label{P-a:objective}\\
	\text{s.t.}\ 
	& 0\leq v_{n,p}, \quad \forall\ \! n \in \mathcal{N},\ p \in \mathcal{P}, \label{P-a:dis}\\
	& \sum_{p=1}^{N_tN_r}v_{n,p} \leq 1, \quad \forall\ \! n \in \mathcal{N}, \label{P-a:g_row}
	\setlength\belowdisplayskip{3pt}
\end{align}
where $\mathcal{P} \triangleq \{1,\cdots,N_tN_r\}$ denotes the index set of the entries of $\boldsymbol{V}$.

Although the (P1-a-$r$) is still a non-convex problem due to its non-concave objective function~\cite{Upperbound}, we can still apply the interior-point algorithm to solve it, which may converge to a locally optimal solution.
To this end, the constraints in~(\ref{P-a:dis}) and~(\ref{P-a:g_row}) can serve as the barrier of the objective function~\cite{Boyd}, based on which (P1-a-$r$) becomes an unconstrained optimization problem, i.e.,
\begin{align}
	\setlength\abovedisplayskip{3pt}
	\label{uncon}
	\max\limits_{\boldsymbol{V}}\  t\log\det\left(\boldsymbol{H}^T\boldsymbol{K}^{-1}\boldsymbol{H}\right) +\boldsymbol{1}^T \log\text{vec}\left(\boldsymbol{V}\right) 
	+ \boldsymbol{1}^T \log\text{vec}\left(\boldsymbol{1}_{N\times1}-\boldsymbol{V}\boldsymbol{1}_{N_tN_r\times1}\right),
	\setlength\belowdisplayskip{3pt}
\end{align}
where $t$ denotes a hyperparameter that changes during the iteration.
To solve~(\ref{uncon}), we first calculate the derivative of $\log\det(\boldsymbol{H}^T\boldsymbol{K}^{-1}\boldsymbol{H})$ w.r.t. $\boldsymbol{V}$, i.e., $\nabla_{\boldsymbol{V}}\log\det(\boldsymbol{H}^T\boldsymbol{K}^{-1}\boldsymbol{H})$, which can be obtained by deriving its matrix differential w.r.t. an intermediate variable $\boldsymbol{H}$ as
\begin{align}
	\setlength\abovedisplayskip{3pt}
	\label{Eq:MIMO_derivative}
	d\left( \log\det(\boldsymbol{H}^T\boldsymbol{K}^{-1}\boldsymbol{H}) \right) &\overset{(a)}{=}\text{Tr}\{\left(\boldsymbol{H}^T\boldsymbol{K}^{-1}\boldsymbol{H}\right)^{-1}d\left(\boldsymbol{H}^T\boldsymbol{K}^{-1}\boldsymbol{H}\right)\} \notag\\
	&\overset{(b)}{=}\text{Tr}\left\{\left(\boldsymbol{H}^T\boldsymbol{K}^{-1}\boldsymbol{H}\right)^{-1}\left( d\boldsymbol{H}^T\boldsymbol{K}^{-1}\boldsymbol{H} + \boldsymbol{H}^T\boldsymbol{K}^{-1}d\boldsymbol{H} \right)\right\} \notag\\
	&\overset{(c)}{=}2\text{Tr}\left\{ \left(\boldsymbol{H}^T\boldsymbol{K}^{-1}\boldsymbol{H}\right)^{-1}\boldsymbol{H}^T\boldsymbol{K}^{-1}d\boldsymbol{H}\right\},
	\setlength\belowdisplayskip{3pt}
\end{align}
where $(a)$ and $(b)$ exploit the differential property of log-determinant and quadratic matrices, respectively, and $(c)$ is due to the trace equality $\text{Tr}(\boldsymbol{B}\boldsymbol{C}) = \text{Tr}(\boldsymbol{C}\boldsymbol{B})$.
Thus, we have $\nabla_{\boldsymbol{H}}\log\det(\boldsymbol{H}^T\boldsymbol{K}^{-1}\boldsymbol{H}) = 2\boldsymbol{K}^{-1}\boldsymbol{H}(\boldsymbol{H}^T\boldsymbol{K}^{-1}\boldsymbol{H})^{-1}$ from~(\ref{Eq:MIMO_derivative}) based on $d(\log\det(\boldsymbol{H}^T\boldsymbol{K}^{-1}\boldsymbol{H}))\\ = \text{Tr}\{(\nabla_{\boldsymbol{H}}\log\det(\boldsymbol{H}^T\boldsymbol{K}^{-1}\boldsymbol{H}))^Td\boldsymbol{H}\}$.
By further applying the chain rule, we can obtain
\vspace{-0.1cm}
\begin{align}
	\setlength\abovedisplayskip{3pt}
	\label{Eq:MIMO_VDeriva}
	\nabla_{\boldsymbol{V}}\log\det(\boldsymbol{H}^T\boldsymbol{K}^{-1}\boldsymbol{H}) 
	& = {\frac{\partial \text{vec}\left(\boldsymbol{H}\right)}{\partial \text{vec}\left(\boldsymbol{V}\right)}}^T\frac{\partial \log\det(\boldsymbol{H}^T\boldsymbol{K}^{-1}\boldsymbol{H}) }{\partial \text{vec}\left(\boldsymbol{H}\right)} \notag\\ 
	& = 2\text{diag}\left( \boldsymbol{H}_c \right) \text{vec}\left(\boldsymbol{K}^{-1}\boldsymbol{H}\left(\boldsymbol{H}^T\boldsymbol{K}^{-1}\boldsymbol{H}\right)^{-1}\right).
	\setlength\belowdisplayskip{3pt}
\end{align}
Then, the search direction for updating $\boldsymbol{V}$ in the interior-point algorithm can be obtained based on~(\ref{Eq:MIMO_VDeriva}) and its stepsize can be determined by the backtracking line search~\cite{Boyd}.
By this means, the iteration of the interior-point algorithm will converge to a locally optimal solution, after which we use a larger $t$ and resolve~(\ref{uncon}) until convergence.

\begin{algorithm}[t]
	\caption{Proposed LIP Algorithm}
	\label{solve_P1}
	\textbf{Input:} $\boldsymbol{H}_1$, $\boldsymbol{H}_c$\\
	\textbf{Output:} $\boldsymbol{F}$, $\boldsymbol{G}$
	\begin{algorithmic}[1]
		\STATE Initialize the entries of $\boldsymbol{G}$ and $\boldsymbol{F}$ as \\
		 $\boldsymbol{G}: g_{n,i} \gets 1, \text{if}\ i = \arg\min\limits_{i} d_{n,i}^1$; $g_{n,k} \gets 0, \forall k\neq i$;\\
		 $\boldsymbol{F}: f_{n,j} \gets 1, \text{if}\ j = \arg\min\limits_{j} d_{j,n}^2$; $f_{n,k} \gets 0, \forall k\neq j$;\\
		 
		 \STATE Initialize the entries of $\boldsymbol{V}$ based on~(\ref{Eq:v_vector});
		
		\REPEAT
		\STATE Calculate the search direction based on~(\ref{Eq:MIMO_VDeriva});
		\STATE Calculate the stepsize by backtracking line search;
		\STATE Update $\boldsymbol{V}$ based on the above search direction and stepsize;
		\UNTIL{Convergence}
		\STATE  Set $n \gets 1$;
		\REPEAT
		\STATE Calculate $p^*$ based on~(\ref{bestp});
		\STATE Calculate $i^* = \lfloor p^*-1/N \rfloor + 1$ and set $g_{n,i^*}$ to 1;
		\STATE Calculate $j^* = \text{mod}( p^*-1, N ) + 1$ and set $f_{n,j^*}$ to 1;
		\STATE Update the index as $n \gets n + 1$;
		\UNTIL{$n > N$}
	\end{algorithmic}
\end{algorithm}
\setlength{\textfloatsep}{0.5cm}

Note that the above process of the interior-point algorithm can be efficiently implemented based on the off-the-shelf $\textit{fmincon}$ function in MATLAB~\cite{lopez2014matlab}.
Nonetheless, the performance and convergence rate of the LIP rely heavily on the initial point~\cite{Boyd}.
According to the Lambertian radiant model, the path loss of the OIRS-reflected channel is determined by the transmitter-OIRS distance and the OIRS-receiver distance, as well as the angles of irradiance and incidence.
Therefore, we propose an initialization scheme based on the prior location information on the LED and PD. 
In particular, we consider that each OIRS reflecting element is initially aligned with the nearest LED and PD, based on which the entries of $\boldsymbol{V}$ can be calculated.

Finally, after the LIP converges, we need to discretize the entries of the optimized $\boldsymbol{V}$ into binary variables, and then retrieve $\boldsymbol{G}$ and $\boldsymbol{F}$ based on~(\ref{Eq:v_vector}).
Let $\boldsymbol{e}_0^{k}$ be an all-zero vector of size $1\times k$ and $\boldsymbol{e}_i^{k} = [0,\cdots,1,\cdots,0]$ be the vector by setting the $i$-th element of $\boldsymbol{e}_0^{k}$ to one.
Moreover, we denote the $n$th row of the optimized $\boldsymbol{V}$ as $\tilde{\boldsymbol{v}}_n \in \mathbb{R}_{1\times N_tN_r}^+$.
Then, for the $n$th OIRS reflecting element, by letting
\begin{equation}
	\setlength\abovedisplayskip{3pt}
	\label{bestp}
	p^* = \arg\min\limits_{p} \left\| \tilde{\boldsymbol{v}}_n - \boldsymbol{e}_p^{N_tN_r} \right\|_2^2,
	\setlength\belowdisplayskip{3pt}
\end{equation}
we can set $g_{n,i^*}=f_{n,j^*}=1$ with $i^*= \lfloor p^*-1/N \rfloor + 1$ and $j^* = \text{mod}( p^*-1, N ) + 1$, respectively. 
The above procedures of the LIP algorithm are summarized in \textbf{Algorithm}~\ref{solve_P1}.

\textbf{2) LDAO algorithm}:
In the LDAO algorithm, we alternately optimize $\boldsymbol{G}$ and $\boldsymbol{F}$ with the other being fixed until convergence.
Firstly, we consider the optimization of $\boldsymbol{G}$ with a fixed $\boldsymbol{F}$.
Let $\boldsymbol{F}_n \triangleq \text{diag}([f_{n,1},\cdots,f_{n,N_r}])$ and $\boldsymbol{G}_n \triangleq \text{diag}([g_{n,1},\cdots,g_{n,N_t}])$ denote the diagonal matrices with their diagonal elements being the $n$th rows of $\boldsymbol{G}$ and $\boldsymbol{F}$, respectively. 
Then, the OIRS-reflected MIMO channel can be rewritten as
\begin{equation}
	\setlength\abovedisplayskip{3pt}
	\label{Eq:H2_reformulated}
	\boldsymbol{H}_2 = \sum_{n = 1}^{N} \boldsymbol{F}_n \tilde{\boldsymbol{H}}_n \boldsymbol{G}_n,
	\setlength\belowdisplayskip{3pt}
\end{equation}
with $\tilde{\boldsymbol{H}}_n \in \mathbb{R}_{N_r\times N_t}^+$ and $[\tilde{\boldsymbol{H}}_n]_{i,j} = h_{i,n,j}^{(2)}$.
Then, for any given $\boldsymbol{F}$ (or $\boldsymbol{F}_n$'s), the subproblem for optimizing $\boldsymbol{G}$ (or $\boldsymbol{G}_n$'s) can be expressed as
\vspace{-0.3cm}
\begin{align}
	\text{(P1-b):}\ \max\limits_{\begin{subarray}{c}
			\{\boldsymbol{G}_n\}
	\end{subarray}}\  & \log\det\left(\left(\boldsymbol{H}_1+\sum_{n = 1}^{N} \boldsymbol{F}_n \tilde{\boldsymbol{H}}_n \boldsymbol{G}_n\right)^T\boldsymbol{K}^{-1}\left(\boldsymbol{H}_1+\sum_{n = 1}^{N} \boldsymbol{F}_n \tilde{\boldsymbol{H}}_n \boldsymbol{G}_n\right)\right) \label{P1-b:objective}\\
	\text{s.t.}\ 
	& \text{Tr}\left( \boldsymbol{G}_n \right) \leq 1, \quad \forall\ \! n \in \mathcal{N} \label{P1-b:GnSum},\\
	& 0\leq g_{n,i}, \quad \forall\ \! n \in \mathcal{N}, i \in \mathcal{T}, \label{P1-b:g_dis}
\end{align}
which is still intractable since the objective function in~(\ref{P1-b:objective}) is non-concave w.r.t. $\boldsymbol{G}_n$.

Note that $\boldsymbol{K}^{-1}$ in the objective function of (P1-b) is a symmetric positive-definite matrix. 
As such, we can invoke the Cholesky decomposition to decompose it as $\boldsymbol{K}^{-1} = \boldsymbol{S}^T\boldsymbol{S}$, where $\boldsymbol{S}$ is an upper triangular matrix with positive diagonal entries. 
Then, the objective function of (P1-b) can be re-expressed as
\begin{equation}
	\label{logdet}
	\tilde g(\{\boldsymbol{G}_n\}) \triangleq \log\det\left(\boldsymbol{S}\left(\boldsymbol{H}_1+\sum_{n = 1}^{N} \boldsymbol{F}_n \tilde{\boldsymbol{H}}_n \boldsymbol{G}_n\right)\right)^T\left(\boldsymbol{S}\left(\boldsymbol{H}_1+\sum_{n = 1}^{N} \boldsymbol{F}_n \tilde{\boldsymbol{H}}_n \boldsymbol{G}_n\right)\right).
\end{equation} 
Although~(\ref{logdet}) is still non-concave in $\boldsymbol{G}_n$~\cite{Upperbound}, it is expressed as a log-det function of a Gram matrix, which can be successively approximated by a quadratic upper bound~\cite{fazel2003log}, as presented in the following lemma.

\vspace{-0.2cm}
\begin{lemma}
	\label{Lemma:upper bound}
	 For any matrices $\tilde{\boldsymbol{A}}$ and $\tilde{\boldsymbol{B}} \in \mathbb{R}_{N_r\times N_t}$, we have the following log-det inequality~\cite{Upperbound}
	\begin{align}
		\label{Eq:upper bound}
		\!\!\! \log\det\left(\tilde{\boldsymbol{A}}^T\tilde{\boldsymbol{A}} + \Delta \boldsymbol{I}_{N_t}\right) \leq  \log\det\left(\tilde{\boldsymbol{B}}^T\tilde{\boldsymbol{B}} + \Delta \boldsymbol{I}_{N_t}\right) + \text{Tr}\left( \left(\tilde{\boldsymbol{B}}^T\tilde{\boldsymbol{B}} + \Delta \boldsymbol{I}_{N_t}\right)^{-1}\tilde{\boldsymbol{A}}^T\tilde{\boldsymbol{A}} \right) + U,
	\end{align}
	where $\Delta$ denotes a small quantity and $U$ denotes a constant.
\end{lemma}

Based on Lemma~\ref{Lemma:upper bound}, we can consider the upper bound on~(\ref{logdet}) as follows. 

\vspace{-0.2cm}
\begin{proposition}
	\label{Prop:upper}
	For any given local points $\tilde{\boldsymbol{G}}_n$, $n \in N$, let $\tilde{\boldsymbol{B}} = \boldsymbol{S}(\boldsymbol{H}_1 + \sum_{n = 1}^{N} \boldsymbol{F}_n \tilde{\boldsymbol{H}}_n\tilde{\boldsymbol{G}}_n)$.
	Then, the expression in~(\ref{logdet}) can be upper-bounded by
	\begin{align}
		\label{Eq:reexpressed}
		\hat g(\{\boldsymbol{G}_n\}) \triangleq \left\| \sum_{n = 1}^{N} \boldsymbol{S}\boldsymbol{F}_n \tilde{\boldsymbol{H}}_n \boldsymbol{G}_n \left(\tilde{\boldsymbol{B}}^T\tilde{\boldsymbol{B}}\right)^{-1/2} \right\|_F^2 + 2 \sum_{n = 1}^{N} \text{Tr}\left( \left(\tilde{\boldsymbol{B}}^T\tilde{\boldsymbol{B}}\right)^{-1} \boldsymbol{H}_1^T\boldsymbol{K}^{-1}\boldsymbol{F}_n \tilde{\boldsymbol{H}}_n \boldsymbol{G}_n  \right) + U',
	\end{align}
	where $U'$ denotes a constant.
	Moreover, $\hat g(\{\boldsymbol{G}_n\})$ and $\tilde g(\{\boldsymbol{G}_n\})$ have the same first-order properties at $\tilde{\boldsymbol{B}}$.
\end{proposition}

\begin{proof}
	By setting $\tilde{\boldsymbol{A}} = \boldsymbol{S}\boldsymbol{H}$ in~(\ref{Eq:upper bound}) and based on Lemma~\ref{Lemma:upper bound}, the objective function in~(\ref{logdet}) can be upper-bounded by
	\vspace{-0.3cm}
	\begin{align}
		\label{P1-b:objective_transform}
		\tilde g(\{\boldsymbol{G}_n\}) = & \lim_{\Delta \to 0} \log\det\left(\tilde{\boldsymbol{A}}^T\tilde{\boldsymbol{A}} + \Delta \boldsymbol{I}_{N_t}\right) \leq \text{Tr}\left( \left(\tilde{\boldsymbol{B}}^T\tilde{\boldsymbol{B}}\right)^{-1} \boldsymbol{H}^T \boldsymbol{K}^{-1} \boldsymbol{H}\right) + \underbrace{\log\det\left(\tilde{\boldsymbol{B}}^T\tilde{\boldsymbol{B}}\right) + U}_{U_0} \notag\\
		=& \text{Tr}\left( \left(\tilde{\boldsymbol{B}}^T\tilde{\boldsymbol{B}}\right)^{-1} \boldsymbol{H}_2^T\boldsymbol{S}^T\boldsymbol{S}\boldsymbol{H}_2 \right) + \text{Tr}\left( \left(\tilde{\boldsymbol{B}}^T\tilde{\boldsymbol{B}}\right)^{-1} \left(\boldsymbol{H}_1^T\boldsymbol{S}^T\boldsymbol{S}\boldsymbol{H}_2 + \boldsymbol{H}_2^T\boldsymbol{S}^T\boldsymbol{S}\boldsymbol{H}_1 \right) \right) \notag\\
		& + \underbrace{\text{Tr}\left( \left(\tilde{\boldsymbol{B}}^T\tilde{\boldsymbol{B}}\right)^{-1} \boldsymbol{H}_1^T \boldsymbol{K}^{-1} \boldsymbol{H}_1 \right) + U_0}_{U'}.
	\end{align}
	Furthermore, the quadratic (first) term in~(\ref{P1-b:objective_transform}) can be rewritten as
	\begin{align}
		\label{P1-b:objective_transform2}
		\text{Tr}\left( \left(\tilde{\boldsymbol{B}}^T\tilde{\boldsymbol{B}}\right)^{-1} \boldsymbol{H}_2^T\boldsymbol{S}^T\boldsymbol{S}\boldsymbol{H}_2 \right) 
		& = \text{Tr}\left(\left(\boldsymbol{S}\boldsymbol{H}_2\left(\tilde{\boldsymbol{B}}^T\tilde{\boldsymbol{B}}\right)^{-1/2} \right)^T \boldsymbol{S}\boldsymbol{H}_2\left(\tilde{\boldsymbol{B}}^T\tilde{\boldsymbol{B}}\right)^{-1/2} \right) \notag\\
		& = \left\| \sum_{n = 1}^{N} \boldsymbol{S}\boldsymbol{F}_n \tilde{\boldsymbol{H}}_n \boldsymbol{G}_n \left(\tilde{\boldsymbol{B}}^T\tilde{\boldsymbol{B}}\right)^{-1/2} \right\|_F^2.
	\end{align}
%
	While for the second term of~(\ref{P1-b:objective_transform}), it can be verified that
	\begin{align}
		\text{Tr}\left( \left(\tilde{\boldsymbol{B}}^T\tilde{\boldsymbol{B}}\right)^{-1} \boldsymbol{H}_1^T\boldsymbol{S}^T\boldsymbol{S}\boldsymbol{H}_2 \right) = \text{Tr}\left( \left(\tilde{\boldsymbol{B}}^T\tilde{\boldsymbol{B}}\right)^{-1} \boldsymbol{H}_2^T\boldsymbol{S}^T\boldsymbol{S}\boldsymbol{H}_1 \right).
	\end{align}
	Based on the above, Eq.~(\ref{logdet}) can be upper-bounded by~(\ref{Eq:reexpressed}).
\end{proof}

According to Proposition~\ref{Prop:upper}, in (P1-b), we can maximize $\hat g(\{\boldsymbol{G}_n\})$ by iteratively updating the local points $\tilde{\boldsymbol{G}}_n$, $n \in N$ or $\tilde{\boldsymbol{B}}$. 
In particular, for any given $\tilde{\boldsymbol{B}}$, we propose to maximize $\hat g(\{\boldsymbol{G}_n\})$ by alternately optimizing each $\boldsymbol{G}_n$ with $\boldsymbol{G}_i$, $i \neq n$ being fixed. 
Then, the subproblem w.r.t. $\boldsymbol{G}_n$ can be written as
\begin{align}
	\setlength\abovedisplayskip{3pt}
	\text{(P1-b-$n$):}\ \max\limits_{\begin{subarray}{c}
			\boldsymbol{G}_n
	\end{subarray}}\ & \hat g \left( \boldsymbol{G}_n ; \{\boldsymbol{G}_i\}_{i \ne n} \right) \label{P1-b':objective}\\
	\text{s.t.}\ 
	& \text{Tr}\left( \boldsymbol{G}_n \right) \leq 1, \label{P1-b':GnSum}\\
	& 0\leq g_{n,i}, \quad \forall\ \! i \in \mathcal{T}. \label{P1-b':g_dis}
	\setlength\belowdisplayskip{3pt}
\end{align}
For (P1-b-$n$), we can obtain its optimal solution in closed-form as follows.

\begin{proposition}
	\label{Prop:Hessian}
	The globally optimal solution to (P1-b-$n$) is given by
	\begin{equation}
		\setlength\abovedisplayskip{3pt}
		\label{Eq:best result}
		\boldsymbol{G}_{n}^* = \text{diag}\left( \boldsymbol{e}_{i^{\star}}^{N_t} \right),
		\setlength\belowdisplayskip{3pt}
	\end{equation}
	where
	\begin{equation}
		\setlength\abovedisplayskip{3pt}
		\label{Eq:best index}
		i^{\star} = \arg\max\limits_{i \in \{0\}\cup \mathcal{T}} \hat g \left( \text{diag}\left( \boldsymbol{e}_i^{N_t} \right) ; \{\boldsymbol{G}_i\}_{i \ne n} \right).
		\setlength\belowdisplayskip{3pt}
	\end{equation}
\end{proposition}

\begin{algorithm}[t]
	\caption{Proposed LDAO Algorithm}
	\label{solve_P1_LDH}
	\textbf{Input:} $\boldsymbol{H}_1$, $\boldsymbol{H}_c$\\
	\textbf{Output:} $\boldsymbol{F}$, $\boldsymbol{G}$
	\begin{algorithmic}[1]
		\STATE Initialize the entries of $\boldsymbol{G}$ and $\boldsymbol{F}$ as \\
		$\boldsymbol{G}: g_{n,i} \gets 1, \text{if}\ i = \arg\min\limits_{i} d_{n,i}^1$; $g_{n,k} \gets 0, \forall k\neq i$;\\
		$\boldsymbol{F}: f_{n,j} \gets 1, \text{if}\ j = \arg\min\limits_{j} d_{j,n}^2$; $f_{n,k} \gets 0, \forall k\neq j$;
		
		\REPEAT
		\STATE Set $n \gets 1$;
		\STATE Calculate $\tilde{\boldsymbol{B}} = \boldsymbol{S}(\boldsymbol{H}_1 + \sum_{n = 1}^{N} \boldsymbol{F}_n \tilde{\boldsymbol{H}}_n\tilde{\boldsymbol{G}}_n)$;
		\REPEAT
		\STATE Calculate $\hat g(\{\boldsymbol{G}_n\})$ based on~(\ref{Eq:reexpressed});
		\STATE Calculate the optimal index $i^{\star}$ based on~(\ref{Eq:best index});
		\STATE Calculate the optimal $\boldsymbol{G}_{n}^*$ based on~(\ref{Eq:best result});
		\STATE Update the index: $n \gets n + 1$;
		\UNTIL{$n > N$}
		\STATE Optimize $\boldsymbol{F}$ using the similar steps as in steps $3$-$10$;
		\UNTIL{Convergence}
	\end{algorithmic}
\end{algorithm}
\setlength{\textfloatsep}{0.5cm}

\begin{proof}
	It is noted from~(\ref{Eq:reexpressed}) that the objective function of (P1-b-$n$) is a convex function, making it a concave quadratic programming (CQP) problem. 
	Although such a problem has been proven to be non-polynomial (NP)-hard to solve~\cite{tuy1998convex}, fortunately, we note that its feasible space is a polyhedron determined by~(\ref{P1-b':GnSum}) and~(\ref{P1-b':g_dis}). 
	Hence, the optimal solution to (P1-b-$n$) must be located at one of the corners of the polyhedron.
	To locate the corner, we can calculate the objective values of (P1-b-$n$) at the $N_t$ corners and select the optimal one, thus giving rise to~(\ref{Eq:best result}) and~(\ref{Eq:best index}).
\end{proof}

After $\boldsymbol{G}_{n}$ is optimized based on the above, the optimization of $\boldsymbol{G}_{n+1}$ follows until the convergence is reached.
Next, for the optimization of $\boldsymbol{F}$ with a fixed $\boldsymbol{G}$, we note that the objective function in~(\ref{Eq:reexpressed}) has the same structure w.r.t. $\boldsymbol{F}_n$ and $\boldsymbol{G}_n$. 
As such, the similar procedures to the above can be used to optimize $\boldsymbol{F}_n$'s, for which the details are omitted for brevity.
The main procedures of the LDAO algorithm are summarized in \textbf{Algorithm}~\ref{solve_P1_LDH}, where we apply the same initialization scheme as in \textbf{Algorithm}~\ref{solve_P1}. 
Note that the convergence of \textbf{Algorithm}~\ref{solve_P1_LDH} can be guaranteed since it must generate non-decreasing objective values of (P1).


\vspace{-0.4cm}
\subsection{Proposed Solution to (P2)}
In this subsection, we aim to solve (P2). 
It can be easily shown that (P2) is a convex optimization problem. 
In particular, we show that its globally optimal solution can be obtained by solving the Karush-Kuhn-Tucker (KKT) conditions, which depends on the relationship between $\boldsymbol{A}_{\max}$ and $A_{\text{total}}$.

Firstly, if $\sum_{i = 1}^{N_t}A_{i,\max} \leq A_{\text{total}}$, the constraint~(\ref{P:total_power}) in (P2) will be inactive. 
Since the objective function of (P2), $f_2(A)$, monotonically increases with each $A_i$. 
Its optimal solution can be easily obtained as
\begin{equation}
	\setlength\abovedisplayskip{3pt}
	\label{Eq:MIMO_Aopti_case1}
	A_i^* = A_{i,\max}.
	\setlength\belowdisplayskip{3pt}
\end{equation}

Secondly, if $A_{i,\max}\geq A_\text{total}/N_t, \forall i \in \mathcal{T}$, both~(\ref{Eq:MIMO_power_constraint}) and~(\ref{Eq:MIMO_totalpower_constraint}) will be active. 
In this case, the maximum of $f_2(A)$ can be attained by invoking the Jensen's inequality~\cite{Boyd}, i.e.,
\begin{equation}
	\setlength\abovedisplayskip{3pt}
	\label{Eq:MIMO_Aopti_case2_Jensen}
	\sum_{i = 1}^{N_t}\log A_i \leq N_t\log\frac{\sum_{i = 1}^{N_t} A_i}{N_t} \leq N_t\log \frac{A_\text{total}}{N_t},
	\setlength\belowdisplayskip{3pt}
\end{equation}
where the equality is achieved at
\begin{equation}
	\setlength\abovedisplayskip{3pt}
	\label{Eq:MIMO_Aopti_case2}
	A_i^* = \frac{A_\text{total}}{N_t},\quad \forall i \in \mathcal{T}.
	\setlength\belowdisplayskip{3pt}
\end{equation}

In other cases, the optimal solution to (P2) can be derived by solving the KKT conditions.
Specifically, the Lagrangian function of (P2) is given by
\vspace{-0.3cm}
\begin{align}
	\label{Eq:MIMO_Aopti_case3_Lagrangian}
	\mathcal{L}\left( \boldsymbol{A}, \boldsymbol{\zeta}, \boldsymbol{\nu}, \varpi \right) = - \sum_{i = 1}^{N_t}\log A_i + \sum_{i = 1}^{N_t}\nu_i \left( A_i -  A_{i,\max} \right) - \sum_{i = 1}^{N_t}\zeta_i A_i + \varpi\left( \sum_{i = 1}^{N_t} A_i - A_\text{total} \right),
\end{align}
where $\boldsymbol{\zeta} \triangleq [\zeta_1, \cdots, \zeta_{N_t}]\in \mathbb{R}_{N_t\times 1}^+$, $\boldsymbol{\nu} \triangleq [\nu_1, \cdots, \nu_{N_t}]\in \mathbb{R}_{N_t\times 1}^+$, and $\varpi \in \mathbb{R}$ denote Lagrange multipliers.
Accordingly, the KKT conditions can be expressed as
\vspace{-0.2cm}
\begin{align}
	\setlength\abovedisplayskip{3pt}
	&\frac{\partial \mathcal{L}}{\partial A_i} = -\frac{1}{A_i^*} - \zeta_i^* + \nu_i^* + \varpi^* = 0, \label{Eq:MIMO_Aopti_case3_KKT1} \\
	&\zeta_i^*A_i^* = \nu_i^*(A_i^* - A_{i,\max}) = 0, \label{Eq:MIMO_Aopti_case3_KKT2}\\
	&A_i^*> 0\geq A_i^* - A_{i,\max}, \label{Eq:MIMO_Aopti_case3_KKT4}\\
	&\zeta_i^*\geq0, \nu_i^*\geq0. \label{Eq:MIMO_Aopti_case3_KKT5}
	\setlength\belowdisplayskip{3pt}
\end{align}
Based on~(\ref{Eq:MIMO_Aopti_case3_KKT2}) and~(\ref{Eq:MIMO_Aopti_case3_KKT5}), it is easy to see $\zeta_i^* = 0$. 
Then, by substituting~(\ref{Eq:MIMO_Aopti_case3_KKT1}) into~(\ref{Eq:MIMO_Aopti_case3_KKT2}), we have
\begin{equation}
	\label{Eq:MIMO_Aopti_case3_equation}
	\left(\frac{1}{A_i^*} - \varpi^* \right)\left( A_i^* - A_{i,\max} \right) = 0, \quad \forall i \in \mathcal{T}.
\end{equation}
It follows from~(\ref{Eq:MIMO_Aopti_case3_equation}) that the optimal solution to (P2) is equal to the minimum of $A_{i,\max}$ and $1/\varpi^*$, i.e.,
\begin{equation}
	\setlength\abovedisplayskip{3pt}
	\label{Eq:MIMO_Aopti_case3_res}
	A_i^* = \min\left( A_{i,\max},\ \frac{1}{\varpi^*} \right), \quad \forall i \in \mathcal{T},
	\setlength\belowdisplayskip{3pt}
\end{equation}
where $\varpi^*$ is a constant satisfying $\sum_{i = 1}^{N_t}A_i^* = A_\text{total}$.

Based on the above, the optimal solution to (P2) can be expressed in closed-form as
\begin{equation}
	\setlength\abovedisplayskip{3pt}
	\label{MIMO_Aopti_case3_A}
	A_i^* = 
	\left\{
	\begin{aligned}
		&A_{i,\max},                 & &\text{if}\ \sum_{i = 1}^{N_t}A_i \leq A_{\text{total}}, \\
		&\frac{A_\text{total}}{N_t}, & &\text{if}\  \min\limits_{i} A_{i,\max}\geq\frac{A_\text{total}}{N_t}, \\
		&\min\left( A_{i,\max},\frac{1}{\varpi^*} \right),& &\text{otherwise}.
	\end{aligned}
	\right.
	\setlength\belowdisplayskip{3pt}
\end{equation}

\vspace{-0.2cm}
\subsection{Complexity Analysis of Proposed Algorithms}
\label{subsec:complex}
In this subsection, we compare the complexity of the proposed LIP and LDAO algorithms, where only the complexity of multiplications is considered while that of additions is ignored.
For the LIP algorithm, it involves the calculation of the inverse $(\boldsymbol{H}^T\boldsymbol{K}^{-1}\boldsymbol{H})^{-1}$, which requires $N_r^2N_t + N_t^2N_r$ and $N_t^3$ operations for matrix multiplications and matrix inverse, respectively. 
Thus, the complexity of calculating $\nabla_{\boldsymbol{V}}\log\det(\boldsymbol{H}^T\boldsymbol{K}^{-1}\boldsymbol{H})$ can be expressed as $\mathcal{O}(2N_r^2N_t + 2N_t^2N_r + N_t^3 + NN_tN_r)$ according to~(\ref{Eq:MIMO_VDeriva}).
Let $I_1$ denote the total number of iterations needed in the LIP algorithm. 
Then, the complexity of \textbf{Algorithm}~\ref{solve_P1} is given by $\mathcal{O}(I_1(2N_r^2N_t + 2N_t^2N_r + N_t^3 + NN_tN_r ) + NN_tN_r ) \approx \mathcal{O}(I_1(2N_r^2N_t + 2N_t^2N_r + N_t^3 + NN_tN_r ))$, where $\mathcal{O}(NN_tN_r)$ denotes the complexity of the minimum distance projection in steps $9$-$14$.
While in the LDAO algorithm, it takes $\mathcal{O}(NN_tN_r)$ operations for solving (P1-b-$n$) to obtain $\boldsymbol{G}_{n}^*$ based on~(\ref{Eq:best result}). 
Let $I_2$ denote the total number of iterations needed in the LDAO algorithm. 
Then, the complexity of \textbf{Algorithm}~\ref{solve_P1_LDH} can be expressed as $\mathcal{O}(I_2NN_tN_r)$.
Moreover, (P2) is solved in closed-form following~(\ref{MIMO_Aopti_case3_A}), where $\varpi^*$ can be obtained by the bisection method. 
Let $J$ represent the total number of bisection search. 
Then, the complexity of solving (P2) is given by $\mathcal{O}(JN_t)$.

Based on the above, it is noted that the LDAO algorithm incurs lower complexity than the LIP algorithm since it decouples (P1) into a series of CQP subproblems, which have much fewer variables to be optimized and admit closed-form solutions as shown in~(\ref{Eq:best result}).
Nevertheless, the LIP algorithm is generally superior to the LDAO algorithm in terms of performance as it directly solves the original problem (P0) by introducing an equivalent problem reformulation.
In summary, these two algorithms strike a balance between the performance and complexity, and further validations will be provided in the next section via simulation.

\vspace{-0.4cm}
\section{Numerical Results}
\vspace{-0.1cm}
\label{Sec:Num}
This section provides simulation results to examine the the proposed algorithms for maximizing the capacity of OIRS-assisted MIMO VLC.
We establish a three-dimensional (3D) Cartesian coordinate system, where OIRS with its four corners located at $(0.0\ \text{m},\ 1.0\  \text{m},\ 1.2\  \text{m})$, $(0.0\ \text{m},\ 1.0\  \text{m},\ 2.9\  \text{m})$, $(0.0\ \text{m},\ 7.0\  \text{m},\ 1.2\  \text{m})$, and $(0.0\  \text{m},\ 7.0 \text{m},\ 2.9 \text{m})$, respectively.
The direct LoS channel $\boldsymbol{H}_1$ is normalized as $\|\boldsymbol{H}_1\|_F^2 = 1$, and the OIRS-reflected channel $\boldsymbol{H}_2$ and the AWGN vector $\boldsymbol{z}$ are scaled by the same factor accordingly. 
Other simulation parameters are summarized in Table~\ref{Tab:Parameters}. 

\begin{table*}[t]
	\centering
	\small
	\caption{Simulation parameters}
	\vspace{-0.4cm}
	\begin{tabular}{| l | c || l | c |}
		\hline
		\multicolumn{1}{|c|}{\textbf{Parameters}} &
		\multicolumn{1}{c||}{\textbf{Values}} & \multicolumn{1}{c|}{\textbf{Parameters}} & \multicolumn{1}{c|}{\textbf{Values}} \\
		\hline
		The number of LEDs, $N_t$ & 4 & The Lambertian index, $m$ & $1$ \\
		\hline
		The number of PDs, $N_r$ & 4 & The reflectivity of OIRS, $\gamma$ & $0.9$ \\
		\hline
		The number of OIRS reflecting elements, $N$ & 32 & The optical filter gain, $g_{of}$ & $1$ \\
		\hline
		The size of room & $8\ \text{m} \times 8\ \text{m} \times 3.5\ \text{m}$ & The area of PD, $A_{PD}$ & $1\ \text{cm}^2$ \\
		\hline
		The location of the receiver center & $(2,\ 3.2,\ 1)\ \text{m}$  & The spacing of PD & $0.2$ m \\
		\hline
		The maximum total emission power, $A_{\text{total}}$ & $4$ W & The semi-angle of FoV, $\Phi_0$ & $70\degree$ \\
		\hline
		The maximum emission power, $\boldsymbol{A}_{\max}$ & $(1.6, 1.4, 0.7, 1)$ W & The refractive index of PD, $q$ & $1.5$ \\
		\hline
	\end{tabular}
	\label{Tab:Parameters}
\end{table*}
\setlength{\textfloatsep}{0.4cm}
\vspace{-0.1cm}
In the simulation, we show the performance of the following two proposed schemes:
\begin{enumerate}
	\item
	\textbf{Proposed scheme 1}:
	This scheme optimizes the OIRS element alignment based on \textbf{Algorithm}~\ref{solve_P1}, and the transmitter emission power $\boldsymbol{A}$ is obtained as~(\ref{MIMO_Aopti_case3_A}).
	
	\item
	\textbf{Proposed scheme 2}:
	The OIRS element alignment is optimized based on \textbf{Algorithm}~\ref{solve_P1_LDH}, and the transmitter emission power is obtained as~(\ref{MIMO_Aopti_case3_A}).
\end{enumerate}
Also, we consider the following baseline schemes for performance comparison.
\begin{enumerate}
	
	\item
	\textbf{Uniform power scheme}:
	The OIRS element alignment matrices are optimized by the proposed LDAO algorithm, while the emission power on the $i$th LED is given by
	\begin{equation}
		\setlength\abovedisplayskip{3pt}
		\label{Eq:uniform power}
		A_i = \min \left(\min_j A_{j,\max}, \frac{A_{\text{total}}}{N_t}\right).  
		\setlength\belowdisplayskip{3pt}
	\end{equation}
	
	\item
	\textbf{Distance-based greedy scheme}:
	OIRS element alignments are determined in a greedy manner, i.e., each OIRS reflecting element is aligned with the nearest transmitter and receiver  antennas.
	Meanwhile, the power allocation is set based on~(\ref{Eq:uniform power}).
	
	\item
	\textbf{No-OIRS scheme}:
	The number of OIRS reflecting elements is set to $N = 0$ and the emission power of the $i$th LED is set based on~(\ref{Eq:uniform power}).
\end{enumerate}
The transmit SNR in VLC is given by the ratio of emission power to noise power~\cite{lapidoth2009capacity}, i.e., $\text{SNR} = \sum_{i = 1}^{N_t}A_i/\sigma N_t$ in Cases \uppercase\expandafter{\romannumeral1} and \uppercase\expandafter{\romannumeral2}, and $\text{SNR} = \sum_{i = 1}^{N_t}E_i/\sigma N_t$ in Case \uppercase\expandafter{\romannumeral3}.

\begin{figure}[t]
	\centering
	\begin{minipage}[b]{0.47\linewidth}
		\centering
		\setlength{\abovecaptionskip}{0.cm}
		\includegraphics[width=1\textwidth]{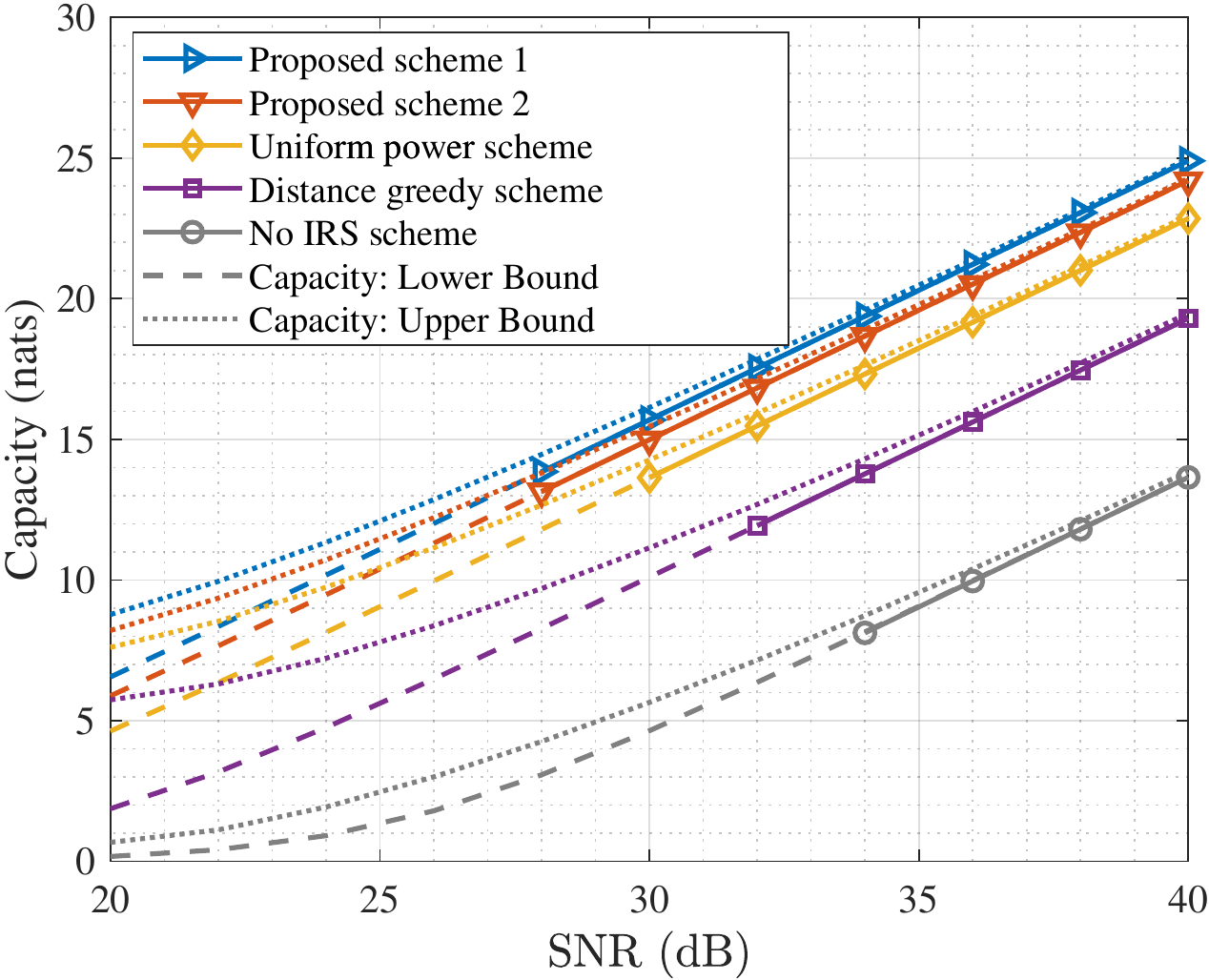}
		\caption{Case \uppercase\expandafter{\romannumeral1}: Capacity of the OIRS-assisted MIMO VLC when $\alpha = 0.4$.}
		\label{Fig:case1}
	\end{minipage}
	\hspace{0.6cm}
	\begin{minipage}[b]{0.47\linewidth}
		\centering
		\setlength{\abovecaptionskip}{0.cm}
		\includegraphics[width=1\textwidth]{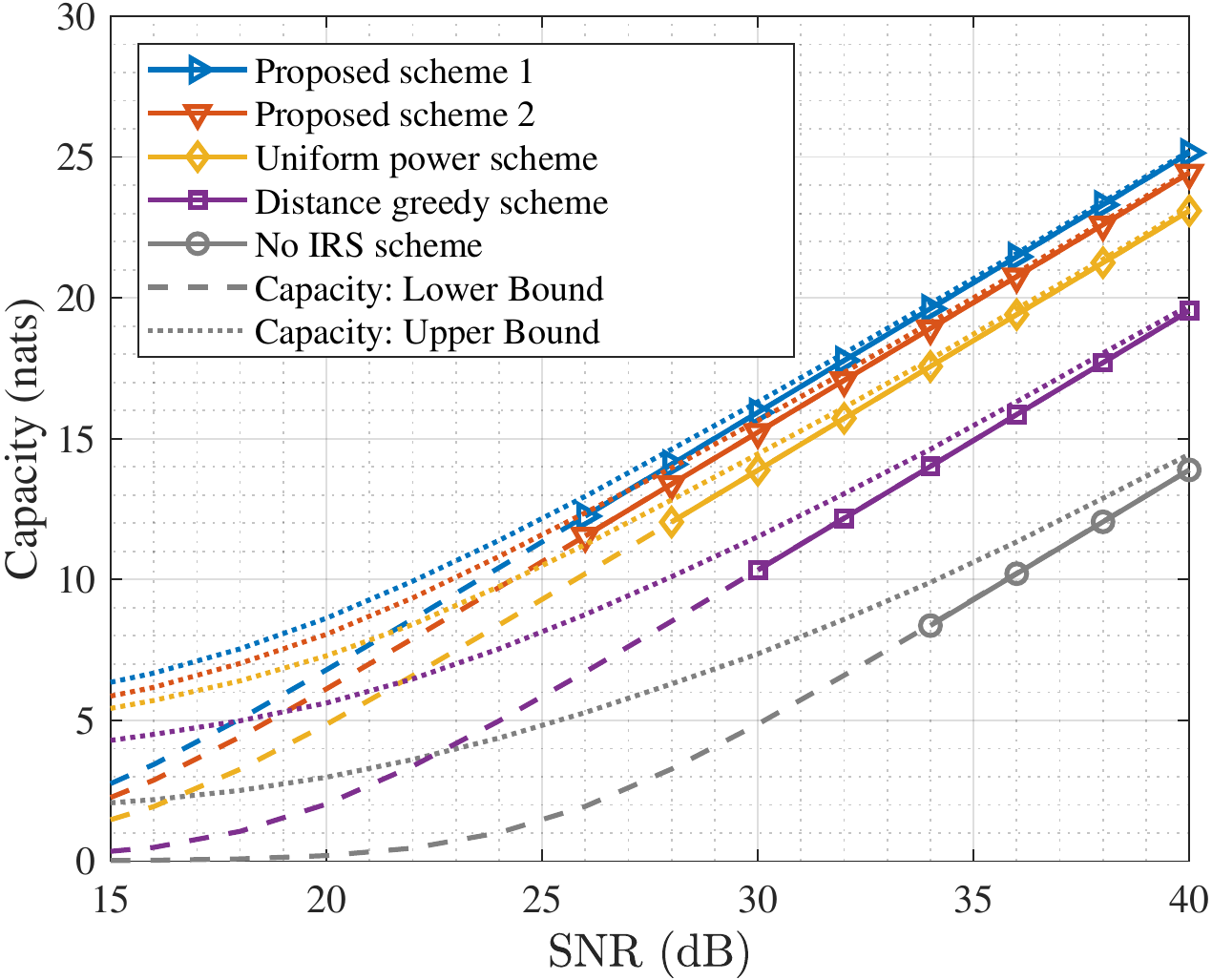}
		\caption{Case \uppercase\expandafter{\romannumeral2}: Capacity of the OIRS-assisted MIMO VLC when $\alpha = 0.5$.}
		\label{Fig:case2}
	\end{minipage}
	\\
	\vspace{0.4cm}
	\begin{minipage}[b]{0.47\linewidth}
		\centering
		\setlength{\abovecaptionskip}{0.cm}
		\includegraphics[width=1\textwidth]{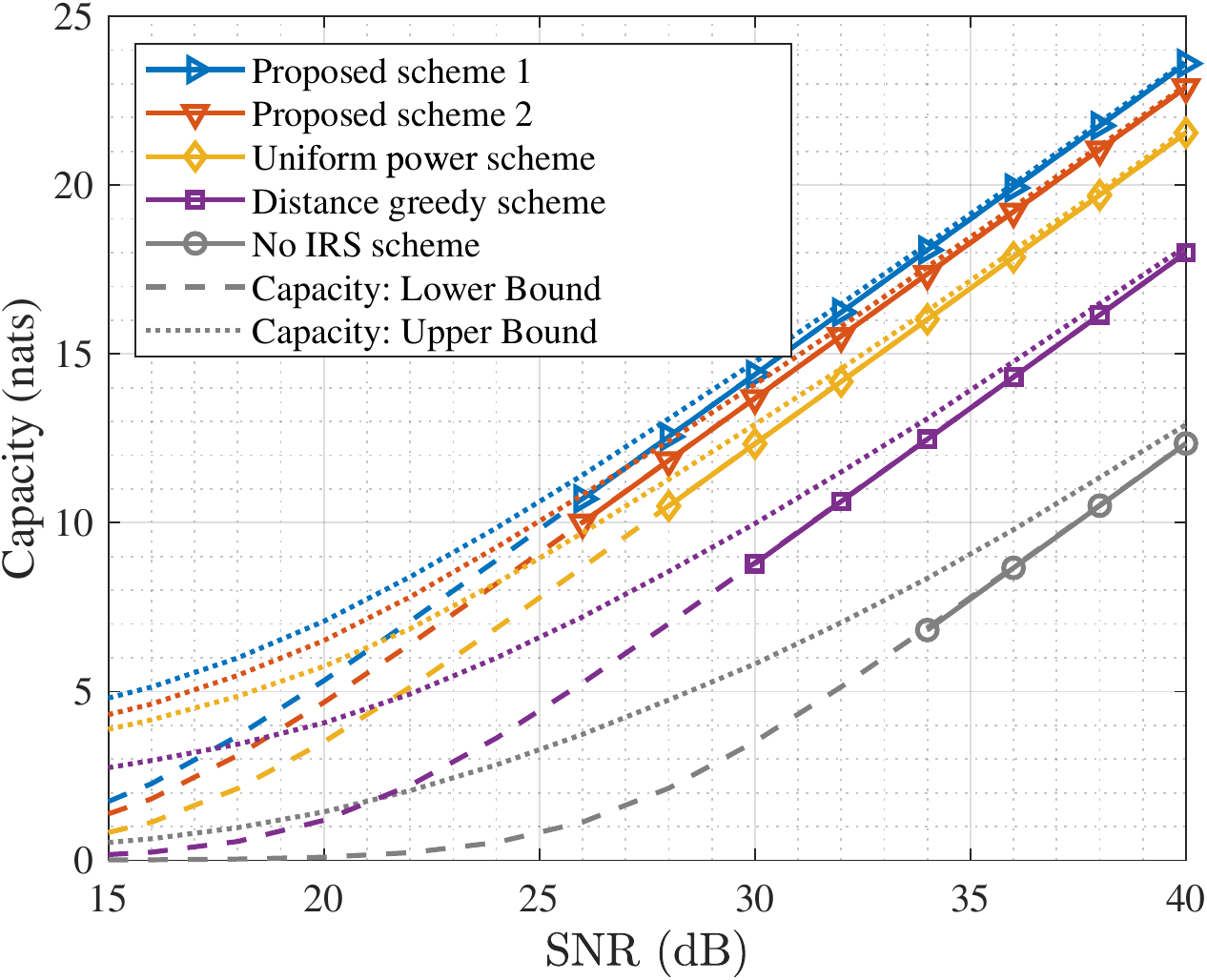}
		\caption{Case \uppercase\expandafter{\romannumeral3}: Capacity of the OIRS-assisted MIMO VLC when $\alpha \to 0$.}
		\label{Fig:case3}
	\end{minipage}
\end{figure}

In Fig.~\ref{Fig:case1}, the capacity performance of the OIRS-assisted MIMO VLC by different schemes is shown under Case \uppercase\expandafter{\romannumeral1}, where the solid line and dashed line denote their corresponding capacity upper bound and lower bound, respectively.
Their exact expressions are given in~(\ref{Eq:MIMOcase1_upperbound}) and~(\ref{Eq:MIMOcase1_EPI}), respectively. 
It is observed that these two bounds converge to the asymptotic capacity expression when SNR grows large, which validates the derived asymptotic capacity in the high-SNR regime.
It is also observed that the performance by all schemes monotonically increases with the SNR. 
In particular, the proposed schemes 1 and 2 are observed to yield much better performance than other baseline schemes.
This is attributed to the optimized transmitter emission power and OIRS element alignment.
Moreover, the performance of scheme 1 is observed to be better than that of scheme 2, which is consistent with our discussion in Section~\ref{subsec:complex}.
Furthermore, it is observed that even the distance-based greedy scheme can outperform the no-OIRS scheme, implying the effectiveness of OIRS in terms of capacity enhancement.

Similarly, the OIRS-assisted MIMO VLC capacities under Cases \uppercase\expandafter{\romannumeral2} and \uppercase\expandafter{\romannumeral3} are depicted in Fig.~\ref{Fig:case2} and Fig.~\ref{Fig:case3}, respectively.
Since the capacity results in these two cases are only different from that in Case \uppercase\expandafter{\romannumeral1} in an offset, as shown in Table~\ref{Tab:capacity}, all observations made for Case \uppercase\expandafter{\romannumeral1} in Fig.~\ref{Fig:case1} are applicable in Fig.~\ref{Fig:case2} and Fig.~\ref{Fig:case3}.
As such, in the sequel of this section, we only present the capacity results in Case \uppercase\expandafter{\romannumeral2} for brevity.

\begin{figure}[t]
	\centering
	\begin{minipage}[b]{0.47\linewidth}
		\centering
		\includegraphics[width=1\textwidth]{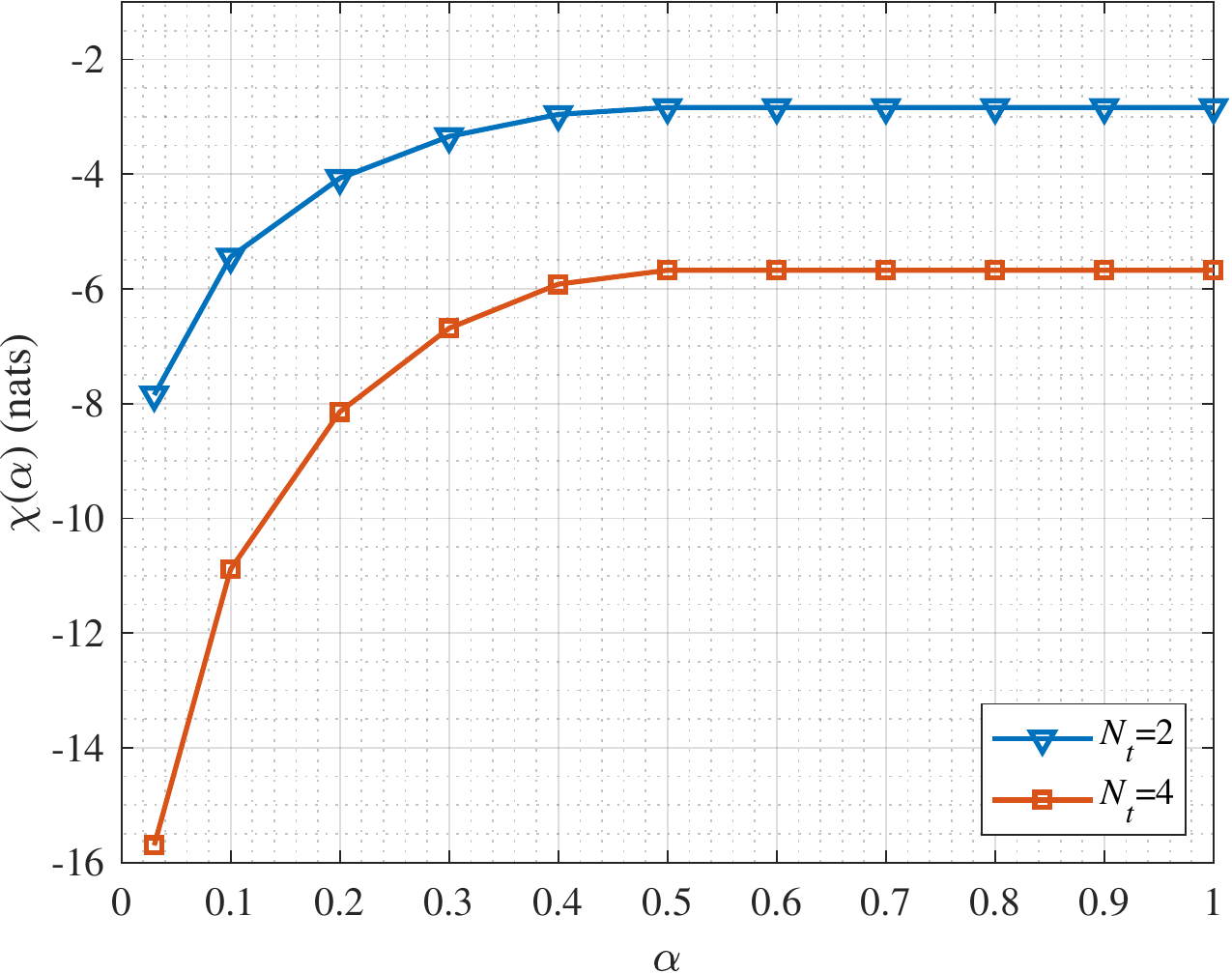}
		\caption{The offset $\chi(\alpha)$ of the OIRS-assisted MIMO VLC capacity for $\alpha \in (0,1]$.}
		\label{Fig:chi}
	\end{minipage}
	\hspace{0.3cm}
	\begin{minipage}[b]{0.47\linewidth}
		\centering
		\includegraphics[width=1\textwidth]{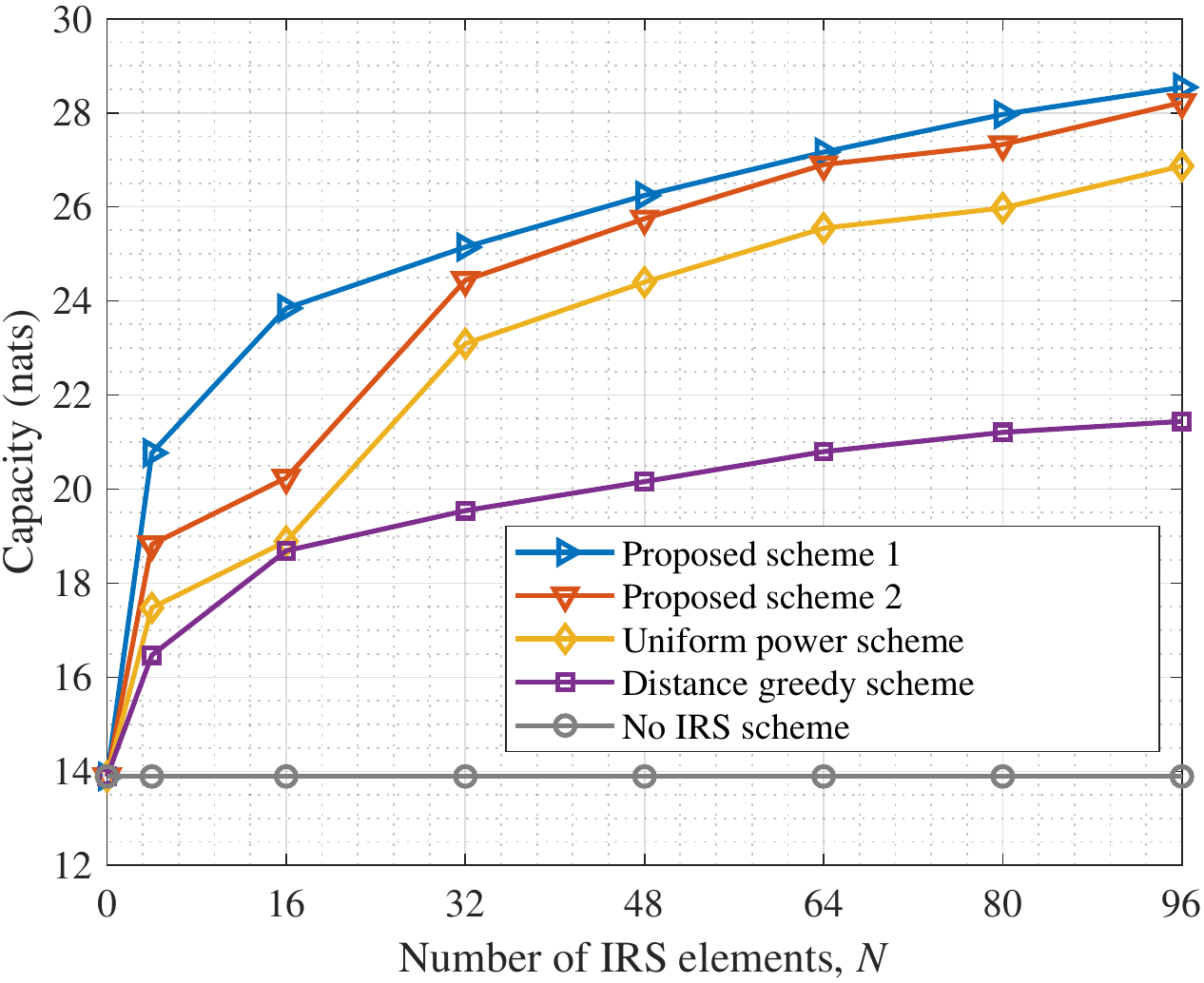}
		\caption{Capacity of the OIRS-assisted MIMO VLC versus the number of OIRS reflecting elements.}
		\label{Fig:number}
	\end{minipage}
\end{figure}


In Fig.~\ref{Fig:chi}, we plot the offset function $\chi(\alpha)$ versus $\alpha$ (see Table~\ref{Tab:capacity}).
It is observed that $\chi(\alpha)$ increases with $\alpha$ when $0<\alpha<1/2$. 
This is expected as for given $A_i$, $i\in \mathcal{T}$, the maximum average power $E_i$ can become larger as $\alpha$ grows, thus resulting in a larger feasible set of input $\boldsymbol{x}$ and hence a larger capacity offset.
Nevertheless, it is observed that $\chi(\alpha)$ remains unchanged when $1/2\leq\alpha\leq1$, which is due to the fact that the average power constraints on $E_i$'s become inactive in this case, as shown via Lemma~\ref{Lemma:thres1/2}.
Moreover, it is also observed that $\chi(\alpha)$ increases with the number of LEDs, $N_t$, as shown in~(\ref{Eq:chi_1}) and~(\ref{Eq:chi_2}).

In Fig.~\ref{Fig:number}, we show the capacity of the OIRS-assisted MIMO VLC versus the number of OIRS reflecting elements with SNR = $40$ dB.
It is observed that the MIMO VLC capacity monotonically increases with $N$, as more transmitter and receiver antenna pairs can be aligned with an OIRS reflecting element, thereby enhancing the received signal power at the receiver. 
As such, it is desired to deploy a sufficiently large OIRS in VLC to boost the system capacity, under the extremely near-field condition.

\vspace{-0.3cm}
\section{Conclusions}
\vspace{-0.2cm}
\label{Sec:Conclude}
In this paper, we characterized the capacity of OIRS-assisted MIMO VLC systems.
Under the extremely near-field condition in VLC, we first developed a new channel model for OIRS and revealed an interesting ``no crosstalk'' property.
Next, we derived the asymptotic capacity expressions for OIRS-assisted MIMO VLC under different emission power constraints in the high-SNR regime and proposed two algorithms to maximize them by jointly optimizing the OIRS element alignment and the transmitter emission power.
Numerical results validated the derived asymptotic capacity and theoretical analysis. 
It was also shown that incorporating an OIRS can significantly enhance the MIMO VLC capacity, and the proposed algorithms can significantly outperform other baseline schemes, thus enabling an appealing solution for capacity enhancement in the next-generation VLC-based wireless systems.

\vspace{-0.3cm}
\begin{appendices}
	\section{Proof of Proposition~\ref{pro1}}
	\label{app:A}
	As the capacity expression of MIMO VLC under a full-rank square channel matrix can be extended directly to that under a full-column-rank channel matrix~\cite{8006585}, without loss of generality, we assume a square and invertible channel matrix $\boldsymbol{H}$, based on which the mutual information can be rewritten as
	\begin{equation}
			\setlength\abovedisplayskip{3pt}
		\label{Eq:MIMO_mutualInfor}
		\textit{I}\left( \boldsymbol{x}; \boldsymbol{Hx} + \boldsymbol{z} \right) = \textit{I}\left( \boldsymbol{x}; \boldsymbol{x}+\boldsymbol{z}'  \right),
			\setlength\belowdisplayskip{3pt}
	\end{equation}
	where $\boldsymbol{z}' \triangleq \boldsymbol{H}^{-1}\boldsymbol{z} \sim \mathcal{N}(\boldsymbol{0}, \boldsymbol{H}^{-1}\boldsymbol{K}\boldsymbol{H}^{-T})$ denotes the equivalent additive white Gaussian noise (AWGN) with the covariance matrix of $\boldsymbol{H}^{-1}\boldsymbol{K}\boldsymbol{H}^{-T}$.
	
	In Case \uppercase\expandafter{\romannumeral1}, it can be proved that the channel capacity is upper-bounded by $\text{sup}_{Q} \mathbb{E} [ \mathscr{D}(W(\boldsymbol{y}|\boldsymbol{x})||R(\boldsymbol{y})) ]$, where $\mathscr{D}(\boldsymbol{x}||\boldsymbol{y})$ and $R(\boldsymbol{y})$ denote the Kullback-Leibler (KL) divergence and the probability measure of the output $\boldsymbol{y}$, respectively.
	Suppose that the distribution of $y_i$ is given by~\cite{lapidoth2009capacity}
	\begin{align}
			\setlength\abovedisplayskip{3pt}
		\label{Eq:Ri}
		R_i(y_i) = \left\{
		\begin{aligned}
			&\frac{1}{\sigma_i\sqrt{2\pi}}e^{-\frac{y_i^2}{2\sigma_i^2}}, && y_i <-\delta_i, \\
			&\frac{\mu^*}{A_i}\frac{\left(1 - 2\mathcal{Q}\left( \frac{
					2\delta_i + A_i}{2\sigma_i}\right)\right)}{\left(e^{\mu^*\frac{\delta_i}{A_i}} - e^{-\mu^*(1+\frac{\delta_i}{A_i})}\right)}e^{-\mu^*\frac{y_i}{A_i}}, && -\delta_i\leq y_i \leq A_i+\delta_i, \\
			&\frac{1}{\sigma_i\sqrt{2\pi}}e^{-\frac{\left(y_i-A_i\right)^2}{2\sigma_i^2}},&& y_i >A_i+\delta_i,
		\end{aligned}
		\right.
			\setlength\belowdisplayskip{3pt}
	\end{align}
	where $\sigma_i^2$ denotes the $i$th diagonal element of covariance matrix of $\boldsymbol{z}'$, $\delta_i \triangleq \sigma_i\log(1 + A_i)$, and $\mu^*$ is given by~(\ref{Eq:MIMOcase1_mu}).
	Therefore, the OIRS-assisted MIMO VLC capacity is upper-bounded by
	\begin{align}
			\setlength\abovedisplayskip{3pt}
		\label{Eq:MIMOcase1_upperbound}
		& C\left( \boldsymbol{G}, \boldsymbol{F}, \boldsymbol{A} \right) \notag\\
		\leq & \sum_{i = 1}^{N_t} \mathop{\text{sup}}\limits_{Q} \mathbb{E} \left[ \mathscr{D}\left(W^{(i)}\left(y_i|x_i\right)||R_i\left(y_i\right)\right) \right] \notag\\
		= & \sum_{i = 1}^{N_t} \Bigg\{ \mathcal{Q}\left(\frac{\delta_i}{\sigma_i}\right) - \frac{1}{2} + \frac{\mu^*\sigma_i}{\sqrt{2\pi}A_i}\left(e^{-\frac{\delta_i^2}{2\sigma_i^2}} - e^{-\frac{(\delta_i + A_i)^2}{2\sigma_i^2}}\right) + \mu^*\alpha\left(1 - 2\mathcal{Q}\left( \frac{
			2\delta_i + A_i}{2\sigma_i}\right)\right)  \notag\\
		& + \frac{\delta_i}{\sigma_i\sqrt{2\pi}}e^{-\frac{\delta_i^2}{2\sigma_i^2}} + \frac{1}{2}\log\sigma_i^2 + \left(1 - 2\mathcal{Q}\left( \frac{
			2\delta_i + A_i}{2\sigma_i}\right)\right)\log\frac{A_i\left(e^{\mu^*\frac{\delta_i}{A_i}} - e^{-\mu^*(1+\frac{\delta_i}{A_i})}\right)}{\sigma_i\sqrt{2\pi}\mu^*\left(1-2\mathcal{Q}\left(\frac{\delta_i}{\sigma_i}\right)\right)} \Bigg\} \notag\\
		& +\frac{1}{2}\log\det\left(\boldsymbol{H}^T\boldsymbol{K}^{-1}\boldsymbol{H}\right),
			\setlength\belowdisplayskip{3pt}
	\end{align}
	where $\mathcal{Q}(\cdot)$ denotes the complementary cumulative distribution function of the standard Gaussian.
	The details on how to derive~(\ref{Eq:MIMOcase1_upperbound}) can be found in~\cite{8006585} and thus are omitted for brevity.
	Moreover, as $\mathcal{Q}(\cdot)$ functions are negligble in the high-SNR regime, based on which we have
	\begin{align}
			\setlength\abovedisplayskip{3pt}
		\label{Eq:MIMOcase1_upperbound_asymptotic}
		\lim\limits_{\forall i\in \mathcal{T},A_i\to \infty} C\left( \boldsymbol{G}, \boldsymbol{F}, \boldsymbol{A} \right)
		\leq& \sum_{i = 1}^{N_t}\log A_i + \frac{1}{2}\log\det\left(\boldsymbol{H}^T\boldsymbol{K}^{-1}\boldsymbol{H}\right) + \chi\left( \alpha \right),
			\setlength\belowdisplayskip{3pt}
	\end{align}
	by taking the limit of $A_i$, $i\in \mathcal{T}$ at both sides of~(\ref{Eq:MIMOcase1_upperbound}).
	
	Conversely, the lower bound on the OIRS-assisted MIMO VLC capacity can be obtained based on the Entropy Power Inequality (EPI)~\cite{ThomasCover}, which is given by
	\begin{align}
		\setlength\abovedisplayskip{3pt}
		\label{Eq:MIMOcase1_EPI}
		\lim\limits_{\forall i\in \mathcal{T},A_i\to \infty} C\left( \boldsymbol{G}, \boldsymbol{F}, \boldsymbol{A} \right) \geq \lim\limits_{\forall i\in \mathcal{T},A_i\to \infty}\frac{N_t}{2}\log\left(1 + e^{\frac{2}{N_t}\left[ \sum_{i = 1}^{N_t}h(x_i) - h(\boldsymbol{z}') \right]}\right) \approx \sum_{i = 1}^{N_t}h(x_i) - h(\boldsymbol{z}'),
		\setlength\belowdisplayskip{3pt}
	\end{align}
	where the approximation is satisfied in the high-SNR regime.
	According to the principle of maximum entropy, given the peak optical intensity and average optical intensity constraints, the largest differential entropy of $x_i$ in~(\ref{Eq:MIMOcase1_EPI}) can be expressed as~\cite{ThomasCover}
	\begin{equation}
		\setlength\abovedisplayskip{3pt}
		\label{Eq:MIMOcase1_maximum entropy}
		h(x_i) = \log\left( \frac{A_i\left(1-e^{-\mu^*}\right)}{\mu^*} \right) + \mu^*\alpha.
		\setlength\belowdisplayskip{3pt}
	\end{equation}
	
	Note that the lower bound given by~(\ref{Eq:MIMOcase1_EPI}) and~(\ref{Eq:MIMOcase1_maximum entropy}) is identical to the right-hand side of~(\ref{Eq:MIMOcase1_upperbound_asymptotic}), which thus gives rise to the asymptotic capacity in~(\ref{Eq:MIMOcase1_capacity}).
	
	\section{Proof of Lemma~\ref{Lemma:thres1/2}}
	\label{app:B}
	The proof of Lemma~\ref{Lemma:thres1/2} is analogous to that in~\cite{lapidoth2009capacity}, except that the peak optical intensity constraint is changed from $x\leq A$ to $\boldsymbol{x} \preceq \boldsymbol{A}$, where $A$ denotes a common peak intensity of all LEDs.
	Define an auxiliary variable as $\bar{\boldsymbol{x}} \triangleq \boldsymbol{A} - \boldsymbol{x}$. 
	Then, it can be inferred that the mutual information between $\boldsymbol{x}$ and $\boldsymbol{y}$ can be rewritten as
	\begin{align}
		\setlength\abovedisplayskip{3pt}
		\label{Eq:MIMO_thres1/2_mutualInfor}
		\textit{I}\left( \boldsymbol{x}; \boldsymbol{Hx} + \boldsymbol{z} \right) = \textit{I}\left( \boldsymbol{A} - \boldsymbol{x}; \boldsymbol{H}\left(\boldsymbol{A} - \boldsymbol{x}\right) + \boldsymbol{z} \right) = \textit{I}\left( \bar{\boldsymbol{x}}; \boldsymbol{H}\bar{\boldsymbol{x}} + \boldsymbol{z} \right).
		\setlength\belowdisplayskip{3pt}
	\end{align}
	Let $B$ denote a Bernoulli variable independent to $\boldsymbol{x}$, with equal probability of $B =1$ and $B =0$.
	Accordingly, we define
	\begin{equation}
		\setlength\abovedisplayskip{5pt}
		\label{xbar}
		\tilde{\boldsymbol{x}} =
		\left\{
		\begin{aligned}
			&\boldsymbol{x}, & &\text{if}\ B=1, \\
			&\bar{\boldsymbol{x}}, & &\text{otherwise}.
		\end{aligned}
		\right.
		\setlength\belowdisplayskip{5pt}
	\end{equation}
	Due to the fact that the entropy is no less than its counterpart conditional entropy, the mutual information is upper-bounded by
	\begin{align}
		\setlength\abovedisplayskip{3pt}
		\label{Eq:MIMO_thres1/2_EPI}
		\textit{I}\!\ ( \tilde{\boldsymbol{x}}; \boldsymbol{H}\tilde{\boldsymbol{x}} + \boldsymbol{z} ) &\geq h\left( \boldsymbol{H}\tilde{\boldsymbol{x}} + \boldsymbol{z} | B \right) - h\left( \boldsymbol{z} \right) \notag\\
		& = \frac{1}{2}\left[ h\left(\boldsymbol{Hx} + \boldsymbol{z}\right) - h\left( \boldsymbol{z} \right) \right] + \frac{1}{2}\left[ h\left(\boldsymbol{H}\bar{\boldsymbol{x}} + \boldsymbol{z}\right) - h\left( \boldsymbol{z} \right) \right] \notag\\
		& = \textit{I}\left( \boldsymbol{x}; \boldsymbol{Hx} + \boldsymbol{z} \right),
		\setlength\belowdisplayskip{3pt}
	\end{align}
	where the two equalities result from~(\ref{xbar}) and~(\ref{Eq:MIMO_thres1/2_mutualInfor}), respectively.
	According to~(\ref{Eq:capacity}), the capacity of the OIRS-assisted MIMO VLC is equal to $\text{sup}\!\ \textit{I}( \boldsymbol{x}; \boldsymbol{Hx} + \boldsymbol{z} )$. 
	Hence, the capacity-achieving distribution of $\boldsymbol{x}$ should be the optimal distribution of $\tilde{\boldsymbol{x}}$, whose expectation is given by
	\begin{equation}
		\mathbb{E}\left[\tilde{\boldsymbol{x}}\right] = \frac{\boldsymbol{A}}{2}.
	\end{equation}
	As a result, the total average optical intensity constraint in~(\ref{Eq:total_average constriant}) becomes
	\begin{equation}
		\label{replaceAverage}
		\sum_{i = 1}^{N_t}\mathbb{E}\left( x_i \right) < \frac{\sum_{i = 1}^{N_t} A_i}{2}
	\end{equation}
	when $E > \sum_{i=1}^{N_t}A_i/2$, thus resulting in the capacity of MIMO VLC in~(\ref{Eq:MIMO_thres1/2_capacity}).
	
	\vspace{-0.4cm}
	
	\section{Proof of Proposition~\ref{pro2}}
	\label{app:C}
	In Case \uppercase\expandafter{\romannumeral2} with $\alpha\in [1/2,1]$, it follows from Lemma~\ref{Lemma:thres1/2} that the capacity-achieving distribution of $\boldsymbol{x}$ satisfies the condition of~(\ref{replaceAverage}). 
	Hence, we only need to consider $\alpha=1/2$ in this case.
	Firstly, the capacity upper bound can be obtained by setting the output distribution as in~\cite{8006585}, yet with different $A_i$ on LEDs.
	Secondly, the capacity lower bound can be obtained by EPI, where the maximum entropy distribution of $\boldsymbol{x}$ under peak optical intensity constraint should follow the uniform distribution, whose probability distribution function can be expressed as
	\begin{equation}
		\setlength\abovedisplayskip{3pt}
		\label{Eq:MIMOcase2_uniformDistri}
		Q\left(\boldsymbol{x}\right) = \prod_{i=1}^{N_t}\frac{1}{A_i}\textbf{I}\{0\leq x_i\leq A_i\}.
		\setlength\belowdisplayskip{3pt}
	\end{equation}
	By this means, we can obtain the asymptotic capacity of the OIRS-assisted MIMO VLC in~(\ref{Eq:MIMOcase2_capacity}), whereas the details are similar to Case \uppercase\expandafter{\romannumeral1} and thus omitted for brevity.
	
\end{appendices}

\bibliographystyle{abbrv}
\vspace{-0.4cm}
\bibliography{IEEEabrv,reference}
\vspace{-0.5cm}
\end{document}